\documentclass[a4paper,greek]{article}

 \usepackage{graphics}
 
 \usepackage{tikz}
\usetikzlibrary{decorations.pathreplacing}
\tikzstyle{every node} = [draw, circle, fill = black, minimum size = 4pt, inner sep = 0pt]
\tikzstyle{normal} = [draw=none, fill = none]

  \usepackage{todonotes}

\usepackage{microtype}

\title{Variants of  Plane Diameter Completion\thanks{The first author was supported by the European Research Council under the European Union's Seventh Framework Programme (FP/2007-2013) / ERC Grant Agreement n. 267959.
The second author was supported by the FP7-PEOPLE-2013-CIG project CountGraph (ref. 630749), the collateral PROCOPE-DAAD project RanConGraph (ref. 57134837), and the Berlin Mathematical School.
The research of the third author was co-financed by the European Union (European Social Fund ESF) and Greek national funds through the Operational Program ``Education and Lifelong Learning'' of the National Strategic Reference Framework (NSRF), Research Funding Program: ARISTEIA~II.}~\thanks{\footnotesize  Emails: \texttt{Petr.Golovach@ii.uib.no}, \texttt{requile@math.fu-berlin.de}, \texttt{sedthilk@thilikos.info}.}}

\author{Petr A. Golovach\thanks{Department of Informatics, University of Bergen, Bergen, Norway.
} \and  
Clément Requilé\thanks{Freie Universit\"at Berlin, Institut f\"ur Mathematik  und Informatik, Berlin, Germany.
}
\and Dimitrios M. Thilikos\thanks{AlGCo project team, CNRS, LIRMM, France, Department of Mathematics,
        University of Athens, 
Athens,        Greece, and Computer Technology Institute \& Press  ``Diophantus'',  Patras, Greece.}
}

\usepackage{amsmath,}
	\usepackage{amsthm}
	\usepackage{amssymb}
	\usepackage{mathrsfs}
	\usepackage{mathdots,stmaryrd}

\newcommand{\intv}[2]{\left \llbracket #1, #2 \right \rrbracket}

\renewcommand{\int}{{\bf int}}

\newcommand{\dist}{{\bf dist}}
\newcommand{\bw}{{\bf bw}}

\newcommand{\diam}{{\bf diam}}

\newcommand{\yes}{\mbox{\sc yes}}

\newcommand{\tw}{{\bf tw}}

\renewcommand{\deg}{\mathrm{deg}}

\newtheorem{theorem}{\bf Theorem}

\newtheorem{lemma}{\bf Lemma}

\newtheorem{corollary}{\bf Corollary}

\newtheorem{proposition}{\bf Proposition}

\date{\empty}
\usepackage{graphicx}

\begin{document}	
\maketitle

\begin{abstract}
\noindent The {\sc Plane Diameter Completion} problem asks, given a plane graph $G$ and a positive integer $d$, if it is a spanning subgraph of a plane graph $H$ that has diameter at most $d$.
We examine two variants of this problem where the input comes with another parameter $k$. In the first variant,  called BPDC, $k$ upper bounds the total number of edges to be added and in the second, called BFPDC, $k$ upper bounds the
number of additional edges per face.
We prove that both problems are {\sf NP}-complete, the first even for 3-connected graphs of face-degree at most 4 and the second even when $k=1$ on 3-connected graphs of face-degree at most 5.
In this paper we give parameterized algorithms for both problems that run in $O(n^{3})+2^{2^{O((kd)^2\log d)}}\cdot n$ steps. 
\end{abstract}

\section{Introduction}

In 1987, Chung \cite[Problem 5]{Chung87} introduced the following problem: find the optimum way to add $q$ edges to a given graph $G$ so that the resulting graph has minimum diameter. (Notice that in all problems defined in this paper we can directly assume that $G$ is a simple graph as loops do not contribute 
to the diameter of a graph and the same holds if we take simple edges instead of multiple ones.) This problem was proved to be {\sf NP}-hard if the aim is to obtain a graph of diameter at most 3~\cite{SchooneBL87}, and later the {\sf NP}-hardness was shown even for the {\sc Diameter-2 Completion} problem~\cite{ChungLMCTSL92}. It is also know that {\sc Diameter-2 Completion} is {\sf W[2]}-hard when parameterized by $q$~\cite{GaoHN13}.

For planar graphs,  Dejter and Fellows introduced in~\cite{DejterF93impr} the {\sc Planar Diameter Completion} problem that asks whether it is possible to obtain a planar graph of diameter at most $d$ from a given planar graph by edge additions. It is not known whether  {\sc Planar Diameter Completion} admits a polynomial time algorithm, but  
 Dejter and Fellows showed that, when parameterized by $d$, {\sc Planar Diameter Completion} is fixed parameter tractable~\cite{DejterF93impr}. The proof is based on the fact that the \yes-instances of the problem are closed under taking minors. Because of the Robertson and Seymour theorem~\cite{RobertsonS-XX} and the algorithm in~\cite{RobertsonS95b}, this implies that, for each $d$, the set of graphs $G$ for which $(G,d)$ is a \yes-instance can be characterized by a {\sl finite} set of forbidden
minors. This fact, along with the minor-checking algorithm in~\cite{RobertsonS95-XIII} implies that there 
exists an $O(f(d)\cdot n^{3})$-step algorithm (i.e. an {\sf FPT}-algorithm) deciding 
whether a plane graph  $G$ has a plane completion of  
diameter at most $d$.  Using the parameterized complexity, this means that 
{\sc Planar Diameter Completion} is {\sf FPT}, when parameterized by $d$.
To make this result constructive, one requires the set of forbidden minors for each $d$,
which is unknown. To find a constructive {\sf FPT}-algorithm for 
this parameterized problem remains a major open problem in parameterized algorithm design.

\medskip\noindent{\bf Our results.}
We denote by $\Bbb{S}_{0}$ the $3$-dimensional  sphere. By a \emph{plane} graph $G$ we
mean a simple planar graph $G$ with the vertex set $V(G)$ and the edge
set $E(G)$ drawn in $\Bbb{S}_0$ such that 
no two edges of this embedding intersect. 
A plane graph $H$ is a {\em a plane completion} (or, simply {\em completion}) of another 
 plane graph $G$ if $H$ is a spanning subgraph of $G$.
A {\em $q$-edge completion} of a plane graph $G$ is a completion $H$ of $G$ where $|E(H)|-|E(G)|\leq q$.
A  {\em $k$-face completion} of a plane graph $G$ is a completion $H$ of $G$ where at most $k$ edges are added in each face of $G$.

In this paper we consider the variants of the {\sc Plane Diameter Completion} problem:

\begin{center}
\fbox{\small\begin{minipage}{12.7cm}
\noindent{\sc Plane Diameter Completion (PDC)}\\
{\sl Input}: a plane graph $G$ and $d\in \mathbb{N}_{\geq 1}$.\\
{\sl Output}: is there a completion of $G$ with diameter  at most $d$?
\end{minipage}}
\end{center}

\noindent
Notice that the important difference between {\sc PDC} and the aforementioned problems is that we consider plane graphs, i.e., the aim is to reduce the diameter of a given embedding of a planar graph preserving the embedding.
In particular we are interested in the following variants:

\begin{center}
\fbox{\small\begin{minipage}{12.7cm}
\noindent{\sc Bounded Budget PDC (BPDC)}\\
{\sl Input}: a plane graph $G$ and $q\in \mathbb{N}, d\in\Bbb{N}_{\geq 1}$ \\
{\sl Question}: is there a completion $H$ of $G$ of diameter  at most $d$ that is also a $q$-edge completion?
\end{minipage}}
\end{center}

\begin{center}
\fbox{\small\begin{minipage}{12.7cm}
\noindent{\sc Bounded Budget/Face PDC (BFPDC)}\\
{\sl Input}: a plane graph $G$ and  $k\in \mathbb{N}, d\in\Bbb{N}_{\geq 1}$. \\
{\sl Question}: is there a completion $H$ of $G$ of diameter  at most $d$ that is also a $k$-face completion?
\end{minipage}}
\end{center}

We examine the complexity of the two above problems. Our hardness results are the following.

\begin{theorem}
\label{thm:NP-c}
Both {\sc BPDC} and {\sc BFPDC}  are {\sf NP}-complete. Moreover,
{\sc BPDC} is {\sc NP}-complete even for 3-connected graphs of face-degree at most 4, and
{\sc BFPDC} is {\sc NP}-complete even for $k=1$ on 3-connected graphs of face-degree at most 5.
\end{theorem}

The  hardness results are proved 
in Section~\ref{npouiwe3d}
using a series of reductions departing  from the \textsc{Planar $3$-Satisfiability} problem that was shown to be {\sf NP}-hard  by Lichtenstein in~\cite{Lichtenstein82}.

The  results of Theorem~\ref{thm:NP-c} prompt us to examine the parameterized complexity\footnote{For more  on parameterized complexity, we refer the reader 
to~\cite{DowneyF13fund}.} of the above problems.
For this, we consider the following  general problem:

\begin{center}
\fbox{\small\begin{minipage}{12.7cm}
\noindent{\sc Bounded Budget  and Budget/Face BDC (BBFPDC)}\\
{\sl Input}: a plane graph $G$,  $q\in\Bbb{N}\cup\{\infty\}$, $k\in \mathbb{N}$, and  $d\in\Bbb{N}_{\geq 1}$. \\
{\sl Question}:  is there a completion $H$ of $G$ of diameter  at most $d$ that is also a $q$-edge completion and a $k$-face completion?
\end{minipage}}
\end{center}
 
Notice that when $q=\infty$ BBFPDC yields BFPDC   and when $q=k$ BBFPDC yields BPDC.
Our main result is that BBFPDC is fixed parameter tractable (belongs in the parameterized class {\sf FPT}) when parameterized by $k$ and $d$.

\begin{theorem}
\label{one}
It is possible to construct an  $O(n^{3})+2^{2^{O((kd)\log d)}} \cdot (\alpha(q))^{2}\cdot n$-step
algorithm for {\sc BBFPDC}.
\end{theorem}
(In the above statement and in the rest of this paper   we use the function $\alpha:\Bbb{N}\cup\{\infty\}\rightarrow\Bbb{N}$ such that 
 if $q=\infty$, then $\alpha(q)=1$, otherwise $\alpha(q)= q$.)

The main ideas of the algorithm of Theorem~\ref{one} are the following. 
We first observe that {\sc yes}-instances of {\sc PDC} and all its variants 
have bounded branchwidth (for the definition of branchwidth, see Section~\ref{defls5t}). 
The typical approach in this case is to derive an {\sf FPT}-algorithm by
either expressing the problem in Monadic Second 
Order Logic -- MSOL (using Courcelle's theorem~\cite{Courcelle97}) or to design a dynamic programming 
algorithm for this problem.
However, for completion problems, this is not really plausible as this logic can quantify 
on {\sl existing} edges or vertices of the graph and not on the ``non-existing'' completion edges. 
This also indicates that to  
design a dynamic programming algorithm for such problems is, in general, not
an easy task. 
In this paper we show how to tackle this problem 
for {\sc BBFPDC} (and its special cases {\sc BPDC} 
and {\sc BFPDC}). 
Our approach is to deal with 
the input  $G$ as a part of a more 
complicated graph with $O(k^{2}\cdot n)$ additional edges, 
namely its {\em cylindrical  enhancement} $G'$ (see Section~\ref{rdekfgment} for the definition).
Informally, sufficiently large cylindrical grids are placed inside the faces of $G$ and then internally vertex disjoint paths in these grids can be used to emulate the edges of a solution of the original problem placed inside the corresponding faces. 
Thus, by the enhancement  we reduce  
{\sc BBFPDC} to a new problem  on $G'$ certified by a suitable 
3-partition of the additional edges. 
Roughly, this partition consists of the 1-weighted 
edges that should be added in the completion, the 0-weighted edges that should link these edges 
to the boundary of the face of $G$ where they will be inserted,  and the $\infty$-weighted edges 
that will be the (useless) 
rest of the additional edges. The new problem asks for such 
a partition that simulates a bounded diameter completion. 
The good news is that, as long as the number of edges per face to be added is bounded, which is the case 
for  {\sc BBFPDC}, the new graph $G'$ has still bounded branchwidth and it is possible, in the new instance,
to quantify this 3-partition of the graph $G'$. However, even under these circumstances, 
to express the new  problem in Monadic Second Order Logic is not easy.
%
For these reasons we decided to follow the more technical 
approach of designing 
a dynamic programming algorithm that leads to the (better) complexity bounds of Theorem~\ref{one}.
This algorithm is quite involved 
due to the technicalities of the translation of the {\sc BBFPDC} to the new problem.
It runs on a sphere-cut decomposition of the plane embedding of $G'$ and its tables encode 
how a partial solution is behaving inside a closed disk whose boundary meets only 
(a few of) the edges of $G'$. We stress that this encoding
takes into account the topological embedding and not just 
the combinatorial structure of $G'$.
Sphere-cut decompositions as well as some 
necessary combinatorial structures for this encoding are 
presented in Section~\ref{jhf5tty}. The dynamic programming  algorithms is presented 
in Section~\ref{dmpop} and  is the most technical part of this paper.

\section{Definitions and preliminaries}
\label{defls5t}

Given a graph~$G,$ we denote by $V(G)$ (respectively $E(G)$)  the \emph{set
of vertices} (respectively \emph{edges}) of~$G$.
A graph~$G'$ is a \emph{subgraph} of a graph~$G$ if $V(G')
\subseteq V(G)$ and $E(G') \subseteq E(G),$ and we
denote this by~$G' \subseteq G$. Also, in case $V(G)=V(G')$, we say that $H$ is a {\em spanning subgraph} of $G$.
If $S$ is a set of vertices or a set of  edges of a graph $G,$
the graph $G \setminus S$ is the graph obtained from $G$ after
the removal of the elements of $S.$ If $S$ is a set of edges, we define $G[E]$ as the graph whose vertex 
set consists of the endpoints of the edges of $E$ and whose edge set of $E$.

\medskip\noindent{\bf Distance and diameter.}
Let $G$ be a graph and let ${\sf w}: E(G)\rightarrow \Bbb{N}\cup\{\infty\}$ (${\sf w}$ is a {\em  weighting 
of the edges of $G$}). Given two vertices $x,x'\in V(G)$ we call {\em $(x,x')$-path} every path of $G$ with $x$ and $x'$
as endpoints. We also 
define ${\sf w}$-\dist$_{G}(x,x')=\min\{{\sf w}(E(P))\mid P\mbox{\ is an $(x,x')$-path in $G$}\}$
and  ${\sf w}$-$\diam(G)=\max\{\mbox{\sf {\sf w}-}\dist_{G}(x,y)\mid x,y\in V(G)\}$ (if $G$ is not connected then $\mbox{\sf {\sf w}-}\diam(G)$ is infinite). When the graph is unweighted then we use \dist$_{G}$ and $\diam$ instead 
of ${\sf w}$-\dist$_{G}$ and ${\sf w}$-\diam.

\medskip\noindent{\bf Plane graphs.}  
To simplify notations on plane graphs, we 
 consider a plane graph $G$ as the union of the points of $\Bbb{S}_0$ in its embedding
 corresponding to its vertices
and edges. That way, a subgraph $H$ of $G$ can be seen as a graph
$H$ where $H\subseteq G$. The {\em faces} of a plane graph $G$, are the connected 
components of the set $\Bbb{S}_{0}\setminus G$. A vertex $v$ (an edge $e$ resp.)  of a plane graph $G$ is \emph{incident} to a face $f$ and, vice-versa, $f$ is incident to $v$ (resp. $e$) if $v$ (resp., $e$) lies on the boundary of $f$. Two faces $f_1,f_2$ are \emph{adjacent} if they have a common incident edges.
We denote by $F(G)$ the set of all faces of $G$.
The {\em degree} of a face $f\in F(G)$ is the number of edges incident to $f$ where bridges 
of $G$ count double in this number. The 
{\em face-degree} of $G$ is the maximum degree of a face in $F(G)$.
Given a face $f$ of $G$, we define $B_{G}(f)$ as  the graph 
whose set of points is the boundary of $f$
and whose vertices are the vertices incident to $f$. 

A set $\Delta \subseteq \Bbb{S}_0$ is an open disc
if it is homeomorphic to $\{(x,y):x^2 +y^2<1\}$.
Also, $\Delta$ is a {\em closed disk} of $\Bbb{S}_0$ if it is the closure of some open disk of $\Bbb{S}_0$.

\medskip\noindent{\bf Branch decomposition.}
Given a graph $H$ with $n$ vertices, a branch decomposition of $H$ is a pair $(T,\mu)$, where $T$ is a 
tree with all internal vertices of degree three
and $\mu : L \rightarrow E(H)$ is a bijection from the set of leaves of $T$ to the edges of $H$. For every edge $e$ of $T$, we define the middle set {\bf mid}($e)\subseteq V(H)$ as follows: if 
	 $T\setminus \{e\}$ has two connected components $T_1$ and $T_2$,
and 	 
	 for $i\in \{1,2\}$, let $H_i^e=H[\{\mu(f) : f\in L\cap V(T_i)]$,
	 and set {\bf mid}($e$) = $V(H_1^e)\cap V(H_2^e)$.
	
The width of $(T,\mu)$ is the maximum order of the middle sets over all edges of $T$, i.e. $\max \{|${\bf mid}$(e)| : e\in T\}$. The {\em branchwidth} of $H$ is the minimum width of 
a branch decomposition of $H$ and is denoted by  \bw($H$).

A {\em grid annulus} $\Gamma_{k,h}$ is the graph obtained by the cartesian product 
of a cycle of $k$ vertices and a path of $h$ vertices. We need the following result.

\begin{proposition}[\cite{GuT10impr}]
\label{gompl}
Let $G$ be a planar graph and $k,h$ be integers with $k \geq 3$ and $h \geq 1$. Then $G$ has either a minor isomorphic to  $\Gamma_{k,h}$ or a branch decomposition of width at most $k+2h-2$.
\end{proposition}

An central feature of the {\sc PDC} problem and its variants is that its {\sc yes}-instances have bounded branchwidth.

\begin{lemma}
\label{btw}
There exists a constant $c_{1}$ such that if  $(G,d)$ is 
a {\sc yes}-instance of {\sc  PDC},
then $\bw(G)\leq c_{1}\cdot d$. The same holds for the graphs in the {\sc yes}-instances of  {\sc BPDC}, {\sc BFPDC},
and {\sc BBFPDC}.
\end{lemma}

\begin{proof}
We examine only the case of {\sc PDC} as 
a {\sc yes}-instance of  {\sc BPDC}, {\sc BFPDC},
and {\sc BBFPDC} is also a {\sc yes}-instance of  {\sc PDC}.

 Notice first that 
if  $G$ has a completion of  diameter at most $d$ and  $G'$ is a minor\footnote{A graph $G'$ is a minor of a graph $G$ if it can obtained applying edge contractions to some subgraph of $G$.} of some $G$, then also $G'$ has a completion $H$ of  diameter at most $d$. Notice also that every completion of the grid 
annulus $\Gamma_{r+2,r+2}$ has diameter $>r$, therefore, if $(G,d)$ is a {\sc yes}-instance of {\sc PDC}, then $G$ cannot contain a $\Gamma_{r+2,r+2}$ as a minor. From Proposition~\ref{gompl}, $G$ has branchwidth bounded by a linear function of $d$ and the lemma follows.
\end{proof}

\section{The reduction}
\label{rdekfgment}

\subsection{cylindrical  enhancements}
 \label{y7yu}

\medskip\noindent{\bf Grid-annulus.}
Let $k$ and $r$ be positive integers where $k\in\Bbb{N}_{\geq 3},r\in\Bbb{N}_{\geq 3}$. We define the graph $\Gamma_{k,r}$ as  the
\emph{$(k\times r)$-grid annulus}, which is the Cartesian product of a path of
$k$ vertices and a cycle of  $r$ vertices.
Notice that $\Gamma_{k,r}$ is uniquely embeddable (up to homeomorphism)
 in the plane and  
has exactly two non-square faces (i.e., faces incident to 4 edges)
$f_{1}$ and $f_{2}$ that are incident only with vertices of degree 3.
We call one of the faces $f_{1}$ and $f_{2}$ the {\em interior} of $\Gamma_{k,r}$ 
and 
the other the {\em exterior} of $\Gamma_{k,r}$. We call the  vertices incident to the interior (exterior) of $\Gamma_{k,r}$ {\em base}  ({\em roof\,}) of $\Gamma_{k,r}$.
Given an edge $e$ in the base of $\Gamma_{k,r}$, we define its {\em ceilings} as 
the set of edges of $\Gamma_{k,r}$ that contains $e$ 
and whose dual edges in  $\Gamma^*_{k,r}$ form a minimum length path between 
the duals of the interior and the exterior face of $\Gamma^{*}_{k,r}$.

\medskip\noindent{\bf Cylindrical  enhancement of a plane graph.}
Let $G$ be a plane graph.
We next give the definition of  the graph 
$G^{(k)}$ for $k\in\Bbb{N}_{\geq 3}$.
Let $f_{i}\in F(G)$  and let $C_{1}^{i},\ldots,C_{\rho_i}^{i}$ be the connected components  of $B_{G}(f_{i})$. 
For each  $C_{j}^{i}$, we denote by $\sigma_{j}^{i}$ 
the number of its edges, agreeing that, in this number, bridge edges count twice
and that if $C_{j}^{i}$ consists of only one vertex, then $\sigma_{j}^{i}=1$.
We then add a copy $\Gamma_{j}^{i}$ of $\Gamma_{k,k\cdot\sigma_{j}^{i}}$ in the embedding of $G$
such that $C_{j}^{i}$ is contained in the interior of $\Gamma_{j}^{i}$ and  all $C_{1}^{i},\ldots,C^{i}_{j-1},\ldots,C_{j+1}^{i}\ldots,C_{\rho_i}^{i}$ are contained in the exterior of $\Gamma_{j}^{i}$ (In Figure~\ref{fig:aopoid} the edges of each $\Gamma_{j}^{i}$ are colored red). 
We then add, for each $v\in C_{j}^{i}$, $\kappa(v)\cdot k$ edges (those around the disks $C_{1},\ldots,C_{4}$ in Figure~\ref{fig:aopoid}) from $v$ to the 
base of $\Gamma_{j}^{i},$ where $\kappa(v)$ is the number of connected components in $C_{j}^{i}\setminus v$ (in the trivial case where $C_{j}^{i}$ consists of only one vertex $v$, then $\kappa(v)=1$).
We add these edges in a way that the resulting embedding remains 
plane and no more than a set $V_{v,i,j}$
$ $ of $k$ consecutive vertices of the  base of $C_{j}^{i}$ are connected with the same vertex $v$ of $C_{j}^{i}$; observe that
there is only one way to add edges so to fulfill these restrictions.  
Notice that the set $V_{v,i,j}$ always induces a path $P_{v,i,j}$ in the resulting graph except  in the case where $C^{i}_{j}$ consists of a single vertex $v$ where $V_{v,i,j}$ induces a cycle. In the later case we pick a maximal path in this cycle and we denote it by $P_{v,i,j}$.
In the example of Figure~\ref{fig:aopoid} the $P_{v,i,j}$'s are the bold paths of the innermost cycle of each $\Gamma_{j}^{i}$. 
We apply this enhancement for each 
connected component of the  boundary of each face of $G$ and 
we denote the resulting graph by $R_{G}^{(k)}$.  

We call a face $f_{i}$ of $R_{G}^{(k)}$ non-trivial if $B_{R_{G}^{(k)}}(f_i)$  has more than one connected components $C_{1}^{i},\ldots,C_{\rho_i}^{i}$. Notice that if $f_{i}$ is non-trivial, each $C_{j}^{i}$ is the roof
of some previously added grid-annulus. For each such grid-annulus, let $J_{j}$ be $k$ consecutive vertices of its roof. We 
add inside $f_{i}$ a copy of $\Gamma_{k,k\cdot \rho_{i}}$
such that its base is a subset of $f_{i}$ and let $\{I_{1},\ldots,I_{\rho_{i}}\}$ be a 
partition of its roof in $\rho_{i}$ parts, each consisting of $k$ consecutive base vertices.
In the example of Figure~\ref{fig:aopoid}, the annulus  $\Gamma_{k,k\cdot \rho_{i}}$
is the one with the edges in the middle of the figure and its base is its innermost cycle.
For each $j\in\{1,\ldots,r_{i}\}$ we add $k$ edges  (depicted as the ``interconnecting'' edges in Figure~\ref{fig:aopoid}) each connecting a vertex  of $J_{j}$ with some  
vertex of $I_{j}$ in a way that the resulting embedding remains plane (again, there is a unique way 
for this to be done). We apply this enhancement for each non-trivial face of $R_{G}^{(k)}$
and we denote the resulting graph  by  $G^{(k)}$. Notice that $G^{(k)}$ is not uniquely defined 
as its definition depends on the choice of the sets $J_{j}$. From now on, we 
always consider an arbitrary choice for $G^{(k)}$ and we call $G^{(k)}$ {\em the $k$-th cylindrical  enhancement of $G$}.
By the construction of $G^{(k)}$, it directly follows that  $|V(G^{(k)})|=O(k^2\cdot n)$.
We say that an edge of $G^{(k)}$ is an {\em expansion edge} if it is an edge of $P_{v,i,j}$ for some $i,j,$ and $v\in V(C_{i,j})$.
Also we denote by $\bar{G}^{(k)}$ the graph created by $G^{(k)}$ if we contract all its expansion edges and all their ceilings of the grid-annuli that were added during the construction of $R_{G}^{(k)}$.

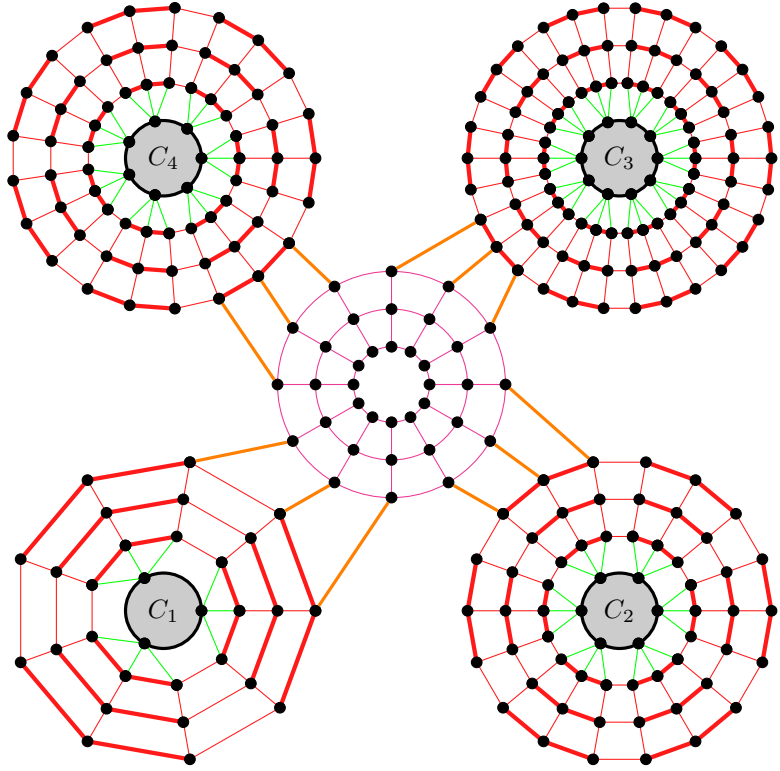
\begin{figure}
\begin{center}
 \begin{tikzpicture}
\begin{scope}[shift={(0,0)}]
\pgfmathsetmacro{\a}{3}
\pgfmathsetmacro{\n}{360/\a}
\pgfmathsetmacro{\nn}{\n/3}
\pgfmathsetmacro{\ne}{\n-\nn}


  \draw[very thick,fill=black!20] (0,0) circle (.5); 
  
 \foreach \a in {0,\n,...,359} {
\foreach \b in {0,\nn,...,\ne}
{
\draw[green!100] (\a:.5) -- (\a-\nn+\b:1);

}

\foreach \b in {\nn,\n,...,\ne}
{
\draw[ultra thick,red!90]  (\a-\nn+\b:1) -- (\a-\nn+\b+\nn:1);
\draw[ultra thick,red!90]  (\a-\nn+\b:1.5) -- (\a-\nn+\b+\nn:1.5);
\draw[ultra thick,red!90]  (\a-\nn+\b:2) -- (\a-\nn+\b+\nn:2);
}
\draw[red!90]  (\a-\nn+\ne:1) -- (\a-\nn+\ne+\nn:1);
\draw[red!90]  (\a-\nn+\ne:1.5) -- (\a-\nn+\ne+\nn:1.5);
\draw[red!90]  (\a-\nn+\ne:2) -- (\a-\nn+\ne+\nn:2);
}

 \foreach \a in {0,\nn,...,359}
  \draw[red!90] (\a:1) -- (\a:2);  
  
 \foreach \a in {0,\n,...,359}
 {
\foreach \b in {0,\nn,...,24}
{\draw (\a+\b:1)  node{};
}
\draw (\a:.5)  node{};
}
 \foreach \a in {0,\nn,...,359}
{
\draw (\a:1.5)  node{};
}
 \foreach \a in {0,\nn,...,359}
 {
\draw (\a:2)  node{};
}
\foreach \a in {0,\n,...,359}
 {
\draw (\a:.5)  node{};
\foreach \b in {0,\nn,...,\ne}
{
\draw (\a+\b:1)  node{};
}
}

\node (a1) at (\nn:2) {};
\node (b1) at (\nn+\n/3:2) {};
\node (c1) at (\nn-\n/3:2) {};

\node[draw=none,fill=none] {$C_{1}$};

\end{scope}

\begin{scope}[shift={(6,0)}]
\pgfmathsetmacro{\a}{6}
\pgfmathsetmacro{\n}{360/\a}
\pgfmathsetmacro{\nn}{\n/3}
\pgfmathsetmacro{\ne}{\n-\nn}

  \draw[very thick,fill=black!20] (0,0) circle (.5); 
  
 \foreach \a in {0,\n,...,359} {
\foreach \b in {0,\nn,...,\ne}
{
\draw[green!100] (\a:.5) -- (\a-\nn+\b:1);

}

\foreach \b in {\nn,\n,...,\ne}
{
\draw[ultra thick,red!90]  (\a-\nn+\b:1) -- (\a-\nn+\b+\nn:1);
\draw[ultra thick,red!90]  (\a-\nn+\b:1.5) -- (\a-\nn+\b+\nn:1.5);
\draw[ultra thick,red!90]  (\a-\nn+\b:2) -- (\a-\nn+\b+\nn:2);
}
\draw[red!90]  (\a-\nn+\ne:1) -- (\a-\nn+\ne+\nn:1);
\draw[red!90]  (\a-\nn+\ne:1.5) -- (\a-\nn+\ne+\nn:1.5);
\draw[red!90]  (\a-\nn+\ne:2) -- (\a-\nn+\ne+\nn:2);
}

 \foreach \a in {0,\nn,...,359}
  \draw[red!90] (\a:1) -- (\a:2);  
  
 \foreach \a in {0,\n,...,359}
 {
\foreach \b in {0,\nn,...,24}
{\draw (\a+\b:1)  node{};
}
\draw (\a:.5)  node{};
}
 \foreach \a in {0,\nn,...,359}
{
\draw (\a:1.5)  node{};
}
 \foreach \a in {0,\nn,...,359}
 {
\draw (\a:2)  node{};
}
\foreach \a in {0,\n,...,359}
 {
\draw (\a:.5)  node{};
\foreach \b in {0,\nn,...,\ne}
{
\draw (\a+\b:1)  node{};
}

}

\node (a2) at (5*\nn+\nn:2) {};
\node (b2) at (5*\nn+\nn+\n/3:2) {};
\node (c2) at (5*\nn+\nn-\n/3:2) {};

\node[draw=none,fill=none] {$C_{2}$};
\end{scope}

\begin{scope}[shift={(6,6)}]
\pgfmathsetmacro{\a}{10}
\pgfmathsetmacro{\n}{360/\a}
\pgfmathsetmacro{\nn}{\n/3}
\pgfmathsetmacro{\ne}{\n-\nn}

  \draw[very thick,fill=black!20] (0,0) circle (.5); 
  
 \foreach \a in {0,\n,...,359} {
\foreach \b in {0,\nn,...,\ne}
{
\draw[green!100] (\a:.5) -- (\a-\nn+\b:1);

}

\foreach \b in {\nn,\n,...,\ne}
{
\draw[ultra thick,red!90]  (\a-\nn+\b:1) -- (\a-\nn+\b+\nn:1);
\draw[ultra thick,red!90]  (\a-\nn+\b:1.5) -- (\a-\nn+\b+\nn:1.5);
\draw[ultra thick,red!90]  (\a-\nn+\b:2) -- (\a-\nn+\b+\nn:2);
}
\draw[red!90]  (\a-\nn+\ne:1) -- (\a-\nn+\ne+\nn:1);
\draw[red!90]  (\a-\nn+\ne:1.5) -- (\a-\nn+\ne+\nn:1.5);
\draw[red!90]  (\a-\nn+\ne:2) -- (\a-\nn+\ne+\nn:2);
}

 \foreach \a in {0,\nn,...,359}
  \draw[red!90] (\a:1) -- (\a:2);  
  
 \foreach \a in {0,\n,...,359}
 {
\foreach \b in {0,\nn,...,24}
{\draw (\a+\b:1)  node{};
}
\draw (\a:.5)  node{};
}
 \foreach \a in {0,\nn,...,359}
{
\draw (\a:1.5)  node{};
}
 \foreach \a in {0,\nn,...,359}
 {
\draw (\a:2)  node{};
}
\foreach \a in {0,\n,...,359}
 {
\draw (\a:.5)  node{};
\foreach \b in {0,\nn,...,\ne}
{
\draw (\a+\b:1)  node{};
}
}

\node (a3) at (17*\nn+\nn:2) {};
\node (b3) at (17*\nn+\nn+\n/3:2) {};
\node (c3) at (17*\nn+\nn-\n/3:2) {};

\node[draw=none,fill=none] {$C_{3}$};

\end{scope}

\begin{scope}[shift={(0,6)}]
\pgfmathsetmacro{\a}{7}
\pgfmathsetmacro{\n}{360/\a}
\pgfmathsetmacro{\nn}{\n/3}
\pgfmathsetmacro{\ne}{\n-\nn}

  \draw[very thick,fill=black!20] (0,0) circle (.5); 
  
 \foreach \a in {0,\n,...,359} {
\foreach \b in {0,\nn,...,\ne}
{
\draw[green!100] (\a:.5) -- (\a-\nn+\b:1);

}

\foreach \b in {\nn,\n,...,\ne}
{
\draw[ultra thick,red!90]  (\a-\nn+\b:1) -- (\a-\nn+\b+\nn:1);
\draw[ultra thick,red!90]  (\a-\nn+\b:1.5) -- (\a-\nn+\b+\nn:1.5);
\draw[ultra thick,red!90]  (\a-\nn+\b:2) -- (\a-\nn+\b+\nn:2);
}
\draw[red!90]  (\a-\nn+\ne:1) -- (\a-\nn+\ne+\nn:1);
\draw[red!90]  (\a-\nn+\ne:1.5) -- (\a-\nn+\ne+\nn:1.5);
\draw[red!90]  (\a-\nn+\ne:2) -- (\a-\nn+\ne+\nn:2);
}

 \foreach \a in {0,\nn,...,359}
  \draw[red!90] (\a:1) -- (\a:2);  
  
 \foreach \a in {0,\n,...,359}
 {
\foreach \b in {0,\nn,...,24}
{\draw (\a+\b:1)  node{};
}
\draw (\a:.5)  node{};
}
 \foreach \a in {0,\nn,...,359}
{
\draw (\a:1.5)  node{};
}
 \foreach \a in {0,\nn,...,359}
 {
\draw (\a:2)  node{};
}
\foreach \a in {0,\n,...,359}
 {
\draw (\a:.5)  node{};
\foreach \b in {0,\nn,...,\ne}
{
\draw (\a+\b:1)  node{};
}
}

\node (a4) at (17*\nn+\nn:2) {};
\node (b4) at (17*\nn+\nn+\n/3:2) {};
\node (c4) at (17*\nn+\nn-\n/3:2) {};

\node[draw=none,fill=none] {$C_{4}$};

\end{scope}

\begin{scope}[shift={(3,3)}]
\pgfmathsetmacro{\a}{9}
\pgfmathsetmacro{\n}{360/\a}
\pgfmathsetmacro{\nn}{\n/3}
\pgfmathsetmacro{\ne}{\n-\nn}

\draw[magenta!90]  (0,0) circle (.5); 
\draw[magenta!90] (0,0) circle (1); 
\draw[magenta!90] (0,0) circle (1.5);

\foreach \b in {0,30,...,359}
{\draw[magenta!90] (\b:.5) --  (\b:1.5) ;
}

\foreach \b in {0,30,...,359}
{\draw (\b:1)  node{};
\draw (\b:.5)  node{};

}


\node (cc2) at (0:1.5) {};

\node (bb3) at (30:1.5) {};
\node (aa3) at (60:1.5) {};
\node (cc3) at (90:1.5) {};

\node (bb4) at (120:1.5) {};
\node (aa4) at (150:1.5) {};
\node (cc4) at (180:1.5) {};

\node (bb1) at (210:1.5) {};
\node (aa1) at (240:1.5) {};
\node (cc1) at (-90:1.5) {};

\node (bb2) at (-60:1.5) {};
\node (aa2) at (-30:1.5) {};

\draw[orange,very thick] (a1) -- (aa1);
\draw[orange,very thick] (a2) -- (aa2);
\draw[orange,very thick] (a3) -- (aa3);
\draw[orange,very thick] (a4) -- (aa4);

\draw[orange,very thick] (b1) -- (bb1);
\draw[orange,very thick] (b2) -- (bb2);
\draw[orange,very thick] (b3) -- (bb3);
\draw[orange,very thick] (b4) -- (bb4);

\draw[orange,very thick] (c1) -- (cc1);
\draw[orange,very thick] (c2) -- (cc2);
\draw[orange,very thick] (c3) -- (cc3);
\draw[orange,very thick] (c4) -- (cc4);

\end{scope}
  \end{tikzpicture}
\end{center}
\caption{An example of a cylindrical  enhancement for $k=3$ inside a non-trivial face of a graph with $4$ connected components (i.e., the boundaries of the disks $C_{1},\ldots,C_{4}$).
}
\label{fig:aopoid}
\end{figure}

\medskip\noindent{\bf Primal-dual drawings.} Let $G$ be a connected plane graph. We denote by $D(G)$ the graph obtained 
if we draw $G$ together with its dual so that dual edges are intersecting to a single point and 
then introduce a vertex to each of these intersection points. We recursively define $D^{(k)}(G)$
such that $D^{(0)}(G)=G$ and $D^{(k)}(G)=D^{k-1}(D(G))$ for every $k\geq 1$.  
The next proposition is a direct consequence of \cite[Lemma~4]{KoutsonasT10plan}.
 \begin{proposition}[]
 There exists some constant $c$ such that 
 for every connected plane graph $G$, it holds 
that $\bw(D(G)) \leq 2\cdot \bw(G)$.
 \end{proposition}
 
 \begin{corollary}
 \label{jkjhtijht}
For every connected plane graph $G$  and $k\in\Bbb{N}_{\geq 1}$, it holds 
that $\bw(D^{(k)}(G)) \leq 2^k\cdot \bw(G)$.
\end{corollary}
  
\begin{lemma}  
\label{i0o9o0opi9o}
If $G$ is a connected plane graph and $k\in\Bbb{N}_{\geq 3}$, then $\bar{G}^{(k)}$ is a minor of $D^{(\lceil \log(k+1)\rceil+1)}(G)$. 
\end{lemma}
      
\begin{proof}
Notice first that $\bar{G}^{(3)}$ is a minor of $D^{(3)}(G)$.
It is then enough to observe  that for every $i\geq 3$,
if  $\bar{G}^{(i)}$ is a minor of $D^{(i)}(G)$, then 
$\bar{G}^{(2i+1)}$ is a minor of $D^{(i+1)}(G)$.
\end{proof}

The following lemma indicates that cylindrical  enhancements do not considerably increase the 
branchwidth of a graph.

\begin{lemma}
\label{lowke}
There is a constant $c_2$ such that 
if $G$ is  an $n$-vertex  plane graph and $k\in\Bbb{N}_{\geq 3}$, 
then $G^{(k)}$ is $3$-connected, $\bw(G^{(k)})\leq c_2\cdot k\cdot \bw(G)$. 
\end{lemma}

\begin{proof}
Let $H$ be the graph  created from $G$ if we add a vertex $v_{f}$ to each non trivial face $f$
and for each of the connected components of $B_{G}(f)$, we arbitrarily pick a vertex and make it
adjacent to  $v_{f}$ by a path of $2k$ internal vertices. As the branchwidth of a non-acyclic 
graph is the maximum branchwidth of its connected components, it follows that $\bw(H)=\bw(G)$.
It is also easy to see that $G^{(k)}$ is a minor of $\bar{H}^{(k)}$. From Lemma~\ref{i0o9o0opi9o},  $\bar{H}^{(k)}$ is a minor of $D^{(r)}(H)$, where $r=\lceil \log(k+1)\rceil+1$.
By Corollary~\ref{jkjhtijht},  it follows that $\bw(D^{(r)}(H))\leq 2^{r}\cdot \bw(H)=O( k\cdot \bw(G))$.
\end{proof}

\subsection{\bf Edge colorings of new edges.}

Let $G$ and $H$  be two plane graphs such that $G$ is a subgraph of $H$ and let $q\in\Bbb{N}\cup\{\infty\}$,  $k\in \mathbb{N}$, and $d\in\Bbb{N}_{\geq 1}$.
Given a 3-partition 
${\bf p}=\{E^{0},E^{1},E^{\infty}\}$ of  $E(H)\setminus E(G),$ 
we define  the function ${\sf w}_{{\bf p}}: E(H)\rightarrow \Bbb{N}$ such that 
\begin{eqnarray*}
{\sf w}_{{\bf p}}  & =  & \{(e,1)\mid e\in E(G)\}\cup \{(e,0)\mid E\in E^{0}\}\cup \\
& &  \{(e,1)\mid e\in E^{1}\}\cup\{(e,d+1)\mid E\in E^{\infty}\}.
\end{eqnarray*}
We say that $G$ has {\em $(q,k,d)$-extension} in $H$ if  there is a 3-partition 
${\bf p}=\{E^{0},E^{1},E^{\infty}\}$ of  $E(H)\setminus E(G)$ 
such that  the following conditions hold
	\begin{enumerate}
	\item[A.] There is no path in $H$ with endpoints in $V(G)$ that consists of  edges in $E^{0}$,
	\item[B.] every face $F$ of $G$ contains at most $k$ edges of $E^{1}$,
	\item[C.] $\forall x,y\in V(G), {\sf w}_{\bf p}\mbox{-\bf dist}_{H}(x,y)\leq d$,	and 
	\item[D.] $|E^{1}|\leq q$.
	\end{enumerate}

Given a 3-partition 
${\bf p}=\{E^{0},E^{1},E^{\infty}\}$ of  $E(H)\setminus E(G)$  
we refer to its elements as the {\em 0-edges}, the {\em 1-edges}, 
and the {\em $\infty$-edges} respectively. We also call the edges of $G$ {\em old-edges}.

Our first step towards our algorithm is to reduce {\sc BBFPDC}  to a problem about  $(q,k,d)$-extensions
of $G$.

 Given a plane graph $G$ and an open set $\Lambda$ of $\Bbb{S}_{0}$,
we define $G\langle \Lambda \rangle$ as the graph whose 
edge set consists of  the edges of $G$ that are subsets of $\Lambda$ and 
whose vertex set consists of their endpoints.

 \medskip\noindent{\bf Disjoint paths.} Let $G$ be a graph. We say that two paths in $G$ are  {\em  disjoint} 
if none of the internal vertices of a path is a  vertex of the other.
Given a collection ${\cal P}$ of pairwise disjoint paths of $G$, we define 
$L({\cal P})=\{\{x,y\}\mid \mbox{\ $x$ and $y$ are the endpoints of a path in ${\cal P}$}\}.$

The proofs of the following  proposition can be found in~\cite{ChatzidimitriouGMRTZ15fixe}.

\begin{proposition}
\label{l7yujl}
Let $G$ be a graph $k\in\Bbb{N}_{\geq 1}$ and  let $H$ be a $k$-face completion of $G$.
For every  face $f\in F(G)$, there is  a collection ${\cal P}$ of $k$ disjoint paths in the graph
$G^{(\max\{3,k\})}\langle f\rangle$ 
 such that $E(G\langle f\rangle) =L({\cal P})$.
\end{proposition}

\begin{lemma}
\label{equiv}
Let $G$ be a plane graph, with $q\in\Bbb{N}\cup\{\infty\}$, $k\in\Bbb{N}_{\geq 1}$ and $d\in\Bbb{N}_{\geq 1}$. 
Then $(G,q,k,d)$ is a {\sc yes}-instance of  {\sc BBFPDC} if and only if $G$ has a {\em $(q,k,d)$-extension} in  $G^{(\max\{3,k\})}$.
\end{lemma}

\begin{proof}
Assume first that $(G,q,k,d)$ is a {\sc yes}-instance of  {\sc BBFPDC} and let 
$H$ be a  completion $H$ of $G$ of diameter  at most $d$ that is also a $q$-edge completion and a $k$-face completion. This means that for every $f\in F(G)$, 
the graph $H_{f}=H\langle f\rangle$ contains at most $k$ edges and that the 
graph $H^{\rm new}=\bigcup_{f\in F(H)}H_{f}$ contains at most $q$ edges.
From Proposition~\ref{l7yujl}, there is a collection ${\cal P}_{f}$ of $y_{f}=|E(H_{f})|$ internally disjoint paths 
in $G^{(\max\{3,k\})}$. Let $E^{1}$ be a set of $y=\sum_{f\in F(G)}y_{f}$ edges obtained if, for every $f\in F(G)$,
we pick one edge from each of the paths in ${\cal P}_{f}$.
Let $E^{0}=E(\bigcup_{f\in F(G)}\bigcup_{P\in{\cal P}_{f}})\setminus E^{1}$ and let $E^{\infty}=E(H^{\rm new})\setminus(E^{0}\cup E^{1})$. 
We now observe that ${\bf p}=\{E^{0},E^{1},E^{\infty}\}$ is a  3-partition 
of  $E(H^{\rm new})=E(G^{(\max\{3,k\})})\setminus E(G)$. By its construction, 
${\bf p}$ satisfies conditions 1--4 of the definition of a $(q,k,d)$-extension of $G$ in $G^{(\max\{3,k\})}$
as required.
  
Let now  ${\bf p}=\{E^{0},E^{1},E^{\infty}\}$ is a  3-partition 
of  $E(H^{\rm new})=E(G^{(\max\{3,k\})})\setminus E(G)$ that is a  $(q,k,d)$-extension of $G$ in $G^{(\max\{3,k\})}$. We construct the graph $H$  by removing from  $G^{(\max\{3,k\})}$
all edges in $E^{\infty}$ and then, in the resulting graph, contract all edges in $E^{0}$.
It is easy to observe that $H$ is a completion of $G$ that is 
also an $q$-edge completion and a $k$-face completion
\end{proof}

\section{Structures for dynamic programming }
\label{jhf5tty}

For our dynamic programming algorithm we need a variant  of branchwidth for plane graphs whose middle sets have  additional topological properties.

\medskip\noindent{\bf Sphere-cut decomposition.}
Let $H$ be a plane graph. An arc is a subset $O$ of the plane homeomorphic to a circle
and is called a {\it noose of $H$} if it meets $H$ only in vertices. We also set $V_{O}=V(H)\cap O$.
An {\em arc} of a noose $O$ is a connected  component of $O\setminus V_{O}$ while 
in the trivial case where $V_{O}=\emptyset$, $O$ does not have arcs.
A {\it sphere-cut decomposition} or {\it sc-decomposition} of $H$ is a triple  $(T,\mu,\pi)$
where $(T,\mu)$ is a branch decomposition of $H$ and $\pi$ is a function mapping each $e\in E(T)$ to cyclic orderings of vertices of $H$, such that
for every $e\in E(T)$ there is  a noose $O_e$ of $H$ where the following properties are satisfied.

	\begin{itemize}
	\item $O_e$ meets every face of $H$ at most once,
	\item $H_{1}^{e}$ is contained in one of the closed disks bounded by $O_{e}$ and $H_{2}^{e}$ is contained in the other ($H_{1}^{e}$ and $H_{2}^{e}$ are as in the definition of branch decomposition).
	\item $\pi(e)$ is a cyclic ordering of $V_{O_{e}}$ defined by a clockwise traversal of $O_e$ in the embedding of $H$.
	\end{itemize}

We denote $X_{e}=V_{O_{e}}$ and  we  always assume that its vertices are clockwise enumerated according to $\pi(e)$.
We denote by ${\bf A}_{e}$ the set containing  the arcs of $O_{e}$. 
Also, if $\pi(e)=[a_{1},\ldots,a_{k},a_{1}]$, then  we use the notation ${\bf A}_{e}=\{a_{1,2},a_{2,3},\ldots,a_{k-1,k},a_{k,1}\}$ where the boundary of the arc $a_{i,i+1}$ consists of the vertices $a_{i}$ and $a_{i+1}$.
We  also define $H_{e}^{+}=(V(H),E(H\cup {\bf A}_{e}))$, i.e., $H_{e}^+$ is the embedding occurring if we add in $H$ the 
arcs of $O_{e}$ as edges. A face of $H_{e}^{+}$ is called {\em internal} if it is not incident to an arc in ${\bf A}_{e}$, i.e., it is also a face of $H$.
A face of $H_{e}^{+}$ is {\em marginal} if it is a properly included is some face of $H$.

For our dynamic programming we require to have in hand an optimal sphere-cut decomposition.
This is done combining the main result of~\cite{GuTa08} and~\cite[(5.1)]{SeymourT94call} (see also~\cite{DornPBF10effi}) and is summarized to the following.

\begin{proposition}
\label{makesphere}
There exists an algorithm that, with input a 3-connected plane graph $G$ and  $w\in \Bbb{N}$,  
outputs a sphere-cut decomposition of $G$ of width  at most $w$ or reports that $\bw(G)>w$.
\end{proposition}

Our next step is to define a series of combinatorial structures that are necessary for our dynamic programming.
Given two sets $A$ and $B$ we denote by $A^B$ the set of all functions 
from $B$ to $A$.

\medskip\noindent{\bf $(d,k,q)$-configurations.} 
Given a set $X$ and a non-negative integer $t$, we say that the pair $({\cal X},\chi)$
is a {\em $t$-labeled partition of $X$}  if ${\cal X}$ is a collection of pairwise disjoint non-empty subsets  of $X$ 
and $\chi$ is a function mapping the integers in $\{1,\ldots,|{\cal X}|\}$ to integers in $\{0,\ldots,t\}$.
In case $X=\emptyset$, a {\em $t$-labeled partition} corresponds to the pair 
$\{\emptyset,\varnothing\}$ where $\varnothing$ is the ``empty" function, i.e. the function whose domain is empty. 
Let $X$ and $A$ be two  finite sets. Given $d,k\in\Bbb{N}$ and $q\in\Bbb{N}\cup\{\infty\}$, we define
a {\em $(d,k,q)$-configuration} of $(X,A)$ as a quintuple  $(({\cal X},\chi),({\cal A},\alpha),({\cal F},{\cal E}),\delta,z)$ where 

\begin{enumerate}
\item $({\cal X},\chi)$  is a $1$-labeled partition of $X$,
\item $({\cal A},\alpha)$ is a  $k$-labeled partition of $A$,
\item $({\cal F},{\cal E})$ is a graph (possibly with loops) where ${\cal F}\subseteq \{0,\ldots,d+1\}^{X}$,
\item $\delta\in \{0,\ldots,d+1\}^{X^2}$, and
\item if $q\in\Bbb{N}$, then $z\leq q$, otherwise $z=\infty$.
\end{enumerate}

\medskip\noindent{\bf Fusions and restrictions.} Let $({\cal X}_1,\chi_1)$
and $({\cal X}_2,\chi_2)$ be two $t$-labeled partitions of the sets $X_{1}$ and $X_{2}$ respectively such that ${\cal X}_{i}=\{X_{1}^i,\ldots,X_{\rho_{1}}^{i}\}, i\in\{1,2\}$. We define ${\cal X}_{1}\oplus {\cal X}_{2}$
as follows: if $x,x'\in X_{1}\cup X_{2}$
we say that $x\sim x'$ if there is a set in ${\cal X}_{1}\cup {\cal X}_{2}$
that contains both of them. Let $\sim_{T}$ be the transitive closure of $\sim$. Then 
${\cal X}_{1}\oplus {\cal X}_{2}$ contains the equivalence classes of $\sim_{T}$.
We now define $\chi_{1}\oplus \chi_{2}$  as follows: let ${\cal X}_{1}\oplus {\cal X}_{2}=\{Y_{1},\ldots,Y_{\rho}\}$. 
Then for each $i\in\{1,\ldots,\rho\}$, we define

$\chi_{1}\oplus \chi_{2}(i)=\min\{t,\sum_{{X_{i'}^1\subseteq  {Y}_{i}}}\chi_{1}(i')+\sum_{{X_{i'}^2\subseteq  {Y}_{i}}}\chi_{2}(i')\}.$
 
The {\em fusion} of  the $t$-labeled partitions  $({\cal X}_1,\chi_1)$
and $({\cal X}_2,\chi_2)$ is the pair $({\cal X}_{1}\oplus {\cal X}_{2},\chi_{1}\oplus \chi_{2})$
that is  a $(t+1)$-labeled partition and is denoted by $({\cal X}_1,\chi_1)\oplus ({\cal X}_2,\chi_2)$.
Given a $t$-labeled partition $({\cal X},\chi)$ of a set $X$ and given a subset $X'$ of $X$
we define the {\em restriction} of  $({\cal X},\chi)$ to $X'$ as the $t$-labeled partition $({\cal X}',\chi')$ of $X'$ where  ${\cal X}'=\{X_{i}\cap X'\mid X_{i}\in {\cal X}\}\setminus \{\emptyset\}$
and $\chi'=\{(i,\chi(i))\mid X_{i}\cap X'\neq\emptyset\}$ and we denote it by $({\cal X},\chi)|_{ X'}$. We also define the intersection of $({\cal X},\chi)$ with $X'$ as the  $t$-labeled partition $({\cal X}',\chi')$ where  ${\cal X}'=\{X_{i}\in {\cal X}\mid X_{i}\cap (X\setminus X')\neq\emptyset\}$
and $\chi'=\{(i,\chi(i))\mid X_{i}\cap X''\neq\emptyset\}$ where $X''=\cup_{X_{i}'\in{\cal X}'}X_{i}$ and we denote it by $({\cal X},\chi)\cap { X'}$. Notice that $({\cal X},\chi)|_{ X'}$ and $({\cal X},\chi)\cap { X'}$ are not always the same.

\section{Dynamic programming}
\label{dmpop}

The following result  is the main algorithmic contribution of this paper.

\begin{lemma}
\label{algdp}
There exists an algorithm that, given  $(G,H,q,k,d,D,b)$ as input where  $G$ and $H$ are 
plane graphs such that $G$ is a  subgraph of $H$, $H$ is 3-connected, 
$q\in \Bbb{N}\cup\{\infty\}, k\in\Bbb{N}$, $d\in\Bbb{N}_{\geq 1}$,  $b\in\Bbb{N}$,
and $D=(T,\mu,\pi)$ is a  sphere-cut decomposition  of $H$
with width at most $b$,
decides whether    $G$ has {\em $(q,k,d)$-extension} in $H$ in
 $(\alpha(q))^{2}\cdot 2^{O(b^2\log d)+2^{O(b\log d)}} \cdot n$ steps.
\end{lemma}

\begin{proof}  
We use the notation $E^{\rm old}=E(G)$ and $E^{\rm new}=E(H)\setminus E(G)$,  $V^{\rm old}=V(G)$ and $V^{\rm new}=V(H)\setminus V(G)$.
We 
choose an arbitrary edge $e^*\in E(T)$, subdivide it by adding a new vertex $v_{\text{new}}$ and update $T$ by adding
 a new vertex $r$ adjacent to $v_{\text{new}}$. 
 We then root $T$ at this vertex $r$ and we extend $\mu$ by setting $\mu(r) = \emptyset$.   
 In $T$ we call {\em leaf-edges} all its edges that are 
 incident to its leaves except from the edge $e_{r}=\{r,v_{\rm new}\}$.
 An edge of $T$ that is not a leaf-edge is called {\em internal}. 
 We denote by $L(T)$ the set of the leaf-edges of $T$ and we denote by $I(T)$ the internal edges of $T$. We also call $e_{r}$ {\em root-edge}.
For each $e\in E(T)$, let $T_e$ be the tree of the forest $T\setminus \{e\}$ that does not contain $r$ as a leaf and let $E_e$ be the edges that are images, via $\mu$, of the leaves of $T$ that are also leaves of $T_e$. We denote $H_e=H[E_e]$ and $V_{e}=V(H_{e})$ and observe that $H_{e_r}=H$. For each edge $e\in I(T)$,  we define its children as the two edges that both belong in the connected component of $T\setminus e$
that does not contain the root $r$ and that share a common endpoint with $e$.
Also, for each edge $e\in E(T)$, we 
define $\Delta_{e}$ as the closed disk bounded by $O_{e}$  such that $G\cap \Delta_{e}=H_{e}$. Finally, for each edge $e\in E(T)$, we set $X_{e}={\bf mid}(e)$, $V^{\rm new}_{e}=V_{e}\cap V^{\rm new}$, $V^{\rm old}_{e}=V_{e}\cap V^{\rm old}$, $E^{\rm new}_{e}=E_{e}\cap E^{\rm new}$, and $E^{\rm old}_{e}=E_{e}\cap E^{\rm old}$.

\medskip\noindent{\bf Distance signatures and dependency graphs.}
Let ${\bf p}=\{E^{0}_{e},E^{1}_{e},E^{\infty}_{e}\}$  be a 3-partition of  $E^{\rm new}_{e}$.   
For each vertex $v\in V_{e},$ we define the {\em $(X_{e},{\bf p})$-distance vector} of $v$
as the function $\phi_{v}: X_{e}\rightarrow\{0,\ldots,d+1\}$ such that if $x\in X_{e}$ then $\phi_{v}(x)=\min\{{\sf w}_{\bf p}\mbox{-\bf dist}_{G_{e}}(v,x),d+1\}.$
We 
define the {\em $(e,{\bf p})$-dependency graph} ${\cal G}_{e,{\bf p}}=({\cal F}_{e,{\bf p}},{\cal E}_{e,{\bf p}})$  (that may contain loops)
where   ${\cal F}_{e,{\bf p}}=\{\phi_{v}\mid v\in V_{e}\}$ and
such that two (not necessarily distinct) vertices  $\phi$ and $\phi'$ of ${\cal F}_{e,{\bf p}}$ are connected by an edge in ${\cal E}_{e,{\bf p}}$
if and only if there exist $v,v'\in V_{e}$ such that $\phi=\phi_{v}$, $\phi'=\phi_{v'}$ and  ${\sf w}_{\bf p}\mbox{-\bf dist}_{H_{e}}(v,v')> d$.
Notice that the set $\Phi_{e}=\{{\cal G}_{e,{\bf p}}\mid {\bf p}\mbox{~is a 3-partition of  $E^{\rm new}_{e}$}\}$
has at most $2^{(d+2)^{|X_{e}|}}$ elements because $\{{\cal F}_{e,{\bf p}}\mid {\bf p}\mbox{~is a 3-partition of  $E^{\rm new}_{e}$}\}\subseteq \{0,\ldots,d+1\}^{X_{e}}$  and, to each ${\cal F}_{e,{\bf p}}$, 
assign a unique edge set ${\cal E}_{e,{\bf p}}$.
Intuitively, each ${\cal F}_{e,{\bf p}}$ corresponds to a partition of the elements of $V_{e}$
such that vertices in the same part have the same $(X_{e},{\bf p})$-distance signature.
Moreover the existence of an edge in the $(e,{\bf p})$-dependency graph between two such parts implies that 
they contain vertices, one from each part, whose ${\sf w}_{\bf p}$-distance in $H_{e}$ is bigger than $d$.

\medskip\noindent{\bf The tables.}
Our aim is to give a dynamic programming algorithm running on the  sc-decomposition $T$. For this, we  describe, for each  $e\in E(T)$, a table $\mathfrak{T}(e)$ containing information on partial solutions of the problem for the graph $G_{e}$
in a way that  the table of an edge $e\in E(T)$ can be computed using the tables of the two children of $e$, the size of each table does not depend on $G$  and the final answer can be derived by the table of the root-edge $e_{r}$.

We define the function $\mathfrak{T}$ mapping each $e\in E(T)$ to a collection $\mathfrak{T}(e)$ 
of $(d,k,q)$-configurations of $(X_{e},{\bf A}_{e})$. In particular, $Q=(({\cal X},\chi),({\cal A},\alpha),({\cal F},{\cal E}),\delta,z)$ $\in \mathfrak{T}(e)$ iff   there exists  a 3-partition ${\bf p}=\{E^{0}_{e},E^{1}_{e},E^{\infty}_{e}\}$ of  $E^{\rm new}_{e}$    such that
the following hold:

\begin{enumerate}
\item $C_{1},\ldots,C_{h}$ are the connected components of $(V(H_{e}),E^{0}_{e})$, 
then 
\begin{itemize}
\item ${\cal X}=\{V(C_{1})\cap X_{e},\ldots,V(C_{h})\cap X_{e}\}$
and 
\item $\forall_{i\in\{1,\ldots,h\}}\ \ \chi(i)= 1$ if $C_{i}$ contains  some vertex of $V_{e}^{\rm old}$, otherwise  $\chi(i) = 0$.
\end{itemize}

 {(The pair $({\cal X},\chi)$ encodes the connected components of the $0$-edges that contain vertices of $X_{e}$
and for each of them registers the number (0 or 1) of the vertices in  $V_{e}^{\rm old}$ in them. This 
information is important to control Condition A.)
}

\item ${\cal A}$ is a partition of ${\bf A}_{e}$ such that  two arcs $A,A'\in {\bf A}_{e}$ belong in the same set, say  $A_{i}$ of ${\cal A}$ 
if and only if they are incident to the same marginal face $f_{i}$ of $H_{e}^{+}$. Moreover, for each $i\in\{1,\ldots,|{\cal A}|\}$, 
$\alpha(i)$ is equal to the number of edges in $E^{1}_{e}$
that are inside $f_{i}$.

 {\small(Here $({\cal A},\alpha)$ encodes the ``partial'' faces of the embedding of $G_{e}$ that are inside $\Delta_{e}$. 
To each of them we correspond the number of $1$-edges that they contain in $H_{e}$. This is useful 
in order to guarantee that   during the algorithm, faces that stop being marginal do not contain
more than $k$ 1-edges, as required by Condition B.)

\item   $({\cal F},{\cal E})$ is the $(e,{\bf p})$-dependency graph, i.e., the graph ${\cal G}_{e,{\bf p}}=({\cal F}_{e,{\bf p}},{\cal E}_{e,{\bf p}})$.  }

 {(Recall that ${\cal F}$ is the collection of all the 
different distance vectors of the vertices of $V_{e}$. 
Notice also that  there might be pairs of 
vertices $x,x'\in V_{e}$ whose ${\bf w}_{\bf p}$-distance 
in $G_{e}$ is bigger than $d$. In order for $G$ 
to have a completion of diameter $d$, these two vertices 
should become connected, at some step of the algorithm,  
by paths  passing {\sl outside} $\Delta_{e}$.
To check this possibility, it is enough to know the  distance vectors of $x$ and $x'$ and 
these are encoded in the set ${\cal F}$. 
Moreover the fact that $x$ and $x'$ are still ``far away'' inside $G_{e}$
is certified by the existence  of an edge (or a loop) between their distance vectors in ${\cal F}$.)
}

\item    For each pair $x,x'\in X_{e}$, $\delta(x,x')=\min\{{\sf w}_{\bf p}\mbox{-\bf dist}_{H_{e}}(x,x'),d+1\}$.

 {(This information is complementary to the one stored in ${\cal F}$
and registers the distances of the vertices in $X_{e}$ inside $H_{e}$.
As we will see, ${\cal F}$ and $\delta$ will be used
in order to compute the distance vectors as well as 
their dependencies during the steps of the algorithm. )
}

\item  There is no path in $H_e$ with endpoints in  $V_{e}^{\rm old}$  that consists of  edges in $E^{0}_{e}$.

 {(This ensures that Condition A is satisfied for the current graph $G_{e}$.)
}

\item  Every internal face of $G^+_{e}$ contains at most $k$ edges in $E_{e}^{1}$.

 {(This ensures that  Condition B holds for all the internal faces of $G_{e}$.)
}

\item  $\forall v,v'\in V_{e}$, either ${\sf w}_{\bf p}\mbox{-\bf dist}_{H_{e}}(v,v')\leq d$ or 
there are two vertices $x,x'\in X_{e}$ such that $\phi_{v}(x)+\phi_{v'}(x')\leq d$.

 {(Here we demand that if two vertices $x_1,x_2$ of  
$V_{e}$ are  ``far away'' (have ${\bf w}_{\bf p}$-distance $>d$)
inside $H_{e}$  then they have some chance  to come ``close'' (obtain ${\bf w}_{\bf p}$-distance
$\leq d$)  in the 
final graph, so that Condition C is satisfied. This 
fact is already stored by an edge in ${\cal E}$ between the two distance vectors of 
$x$ and $x'$ and the possibility that $x_1$ and $x_2$ may come close at some step 
of the algorithm, in what concerns the graph $G_{e}$, depends only on these distance vectors 
and not on the vertices $x_1$ and $x_2$ themselves.)
}

 \item  There are at most $z$ edges of $E_{e}^{1}$ inside the internal faces of $G^{+}_{e}$ (clearly, this last condition becomes void when $q=\infty$). 

 {(This information helps us control Condition D during  the algorithm.)}

\end{enumerate}
Notice that in case $X_{e}=\emptyset$ the only graph that can correspond to the 6th step is the graph $(\{\varnothing\},\emptyset)$ which, from now on will be denoted by $G_{\varnothing}$.

\medskip\noindent{\bf Bounding the set of characteristics.}
Our next step is to bound $\mathfrak{T}(e)$ for each $e\in E(T)$. Notice first that $|X_{e}|=|{\bf A}_{e}|\leq b$.
This means that there are $2^{O(b\log b)}$ instantiations of $({\cal X},\chi)$ 
and $2^{O(k+b\log b)}$ instantiations of $({\cal A},\alpha)$. As we previously noticed, the different instantiations 
of $({\cal F},{\cal E})$ are  $|{\Phi}_{e}|= 2^{2^{O(b\log d)}}$. Moreover, there are 
$2^{O(b^2\log d)}$ instantiations of $\delta$ and $\alpha(q)$ instantiations of $z$.
We conclude that there exists a function $f$ such that 
for each $e\in V(T)$, $|\mathfrak{T}(e)|\leq f(k,q,b,d)$. Moreover,  $f(k,q,b,d)=\alpha(q)\cdot 2^{O(b^2\log d)+2^{O(b\log d)}}$.

\medskip\noindent{\bf The characteristic function on the root edge.}
Observe  that $E_{\rm new}$ is $(k,d,q,{\sf w})$-edge colorable
in $H$ if and only if $\mathfrak{T}(e_{r})\neq\emptyset$,~~ i.e., $((\emptyset,\varnothing),(\emptyset,\varnothing),G_{\varnothing},$$\varnothing,z)\in \mathfrak{T}(e_r)$ for some $z\leq q$. Indeed, if this happens, conditions 1--4 become void while  conditions 5, 6, 7, and 8 imply that $H=H_{e}$ satisfies the conditions A, B, C, and D respectively 
in the definition of the $(k,d,q,{\sf w})$-edge colorability of $E^{\rm new}$.

\medskip\noindent{\bf The  computation of the tables.} 
We will now show how to compute  $\mathfrak{T}(e)$ for each $e\in E(T)$.

We now give the definition of $\mathfrak{T}(e)$ in the case where $e$ is a leaf of $T$ is the following:
Given a $q\in\Bbb{N}\cup\{\infty\}$, we define $A(q)=\{\infty\}$ if $q=\infty$, otherwise
$A(q)=\{z\mid z\leq q\}$.

Suppose now that $e_{l}$ is a leaf-edge of $T$ where $\pi(e_{l})=[a_1,a_2,a_1]$ and ${\bf A}_{e_{l}}=\{a_{1,2},a_{2,1}\}$.

\begin{enumerate}
\item If $\{a_{1},a_{2}\}\in E_{e}^{\rm old}$, then
\begin{eqnarray*}
\frak{T}(e_{l}) & =& \{\ \big((\{\{a_1\},\{a_2\}\},\{(1,1),(2,1)\}),\\
& &~~\, (\{\{a_{1,2}\},\{a_{2,1}\}\},\{(1,0),(2,0)\}),\\
& &~~\, \big(\big\{\{(a_{1},0),(a_{2},{\sf w}(\{a_{1},a_{2}\}))\},\{(a_{1},{\sf w}(\{a_{1},a_{2}\})),(a_{2},0)\}\big\},\emptyset\big),\\
& &~~\,  \{((a_{1},a_{2}),{\sf w}(\{a_{1},a_{2}\}))\},z\big)\mid z\in A(q)\},
\end{eqnarray*}

\item if $\{a_{1},a_{2}\}\in E_{e}^{\rm new}$ and $\{a_{1},a_{2}\}\subseteq V_{e}^{\rm old}$, then $\frak{T}(e_{l}) = {\cal Q}^{1}\cup {\cal Q}^{\infty}$  where 
\begin{eqnarray*}
{\cal Q}^{1} & =& \{\  \big(\ (\{\{a_1\},\{a_2\}\},\{(1,1),(2,1)\})       \\
& &~~\,  ( \{\{a_{1,2},a_{2,1}\}\},\{(1,1)\}   )\\
& &~~\,  ( \big\{\{(a_{1},0),(a_{2},1)\},\{(a_{1},1),(a_{2},0)\}\big\},\emptyset)\\
& &~~\,  \{((a_{1},a_{2}),s)\},z \,  \big)\mid z\in A(q)-\{0\} \}\\
\\
{\cal Q}^{\infty} & =&  \{\ \big(\ (\{\{a_1\},\{a_2\}\},\{(1,1),(2,1)\})       \\
& &~~\,  ( \{\{a_{1,2},a_{2,1}\}\},\{(1,0)\}   )\\
& &~~\,  ( \big\{\{(a_{1},0),(a_{2},d+1)\},\{(a_{1},d+1),(a_{2},0)\}\big\},K)\\
& &~~\,  \{((a_{1},a_{2}),d+1)\},z \,  \big)\mid z\in A(q)\}
\end{eqnarray*}
(the set $K$ above contains a single edge that is not a loop),
and if $\{a_{1},a_{2}\}\in E_{e}^{\rm new}$ and $\{a_{1},a_{2}\}\nsubseteq V_{e}^{\rm old}$, then $\frak{T}(e_{l}) = {\cal Q}^{1}\cup {\cal Q}^{\infty}\cup{\cal Q}^{0}$ where 
\begin{eqnarray*}
{\cal Q}^{0} & =&  \{\ \big(\ (\{\{a_{1},a_{2}\}\},\{(1,1-\langle \{a_{1},a_{2}\}\subseteq V_{e}^{\rm new}\rangle)\})        \\
& &~~\,  ( \{\{a_{1,2},a_{2,1}\}\},\{(1,0)\}   )\\
& &~~\,  (\{\{(a_{1},0),(a_{2},0)\}\},\emptyset)\\
& &~~\,  \{((a_{1},a_{2}),0)\},z \,  \big)\mid z\in A(q)\}
\end{eqnarray*}
\end{enumerate}

Assume now that $e$ is a non-leaf edge of $T$ with children $e_{l}$ and $e_{r}$, the collection $\frak{T}(e)$ is given by ${\bf join}(\frak{T}(e_1),\frak{T}(e_2))$ where {\bf join} is a procedure that is depicted below.
 Notice that ${\bf A}_{e}$ is the symmetric difference of ${\bf A}_{e_{l}}$ and ${\bf A}_{e_{r}}$
and $X_{e}$  consists of  the endpoints of the arcs in ${\cal A}_{e}$. We also set $X_{e}^{F}=(X_{e_{l}}\cup X_{e_{r}})\setminus X_{e}$.

\begin{center}
\fbox{\small
\begin{minipage}{12cm}
\begin{tabbing}
Procedure {\bf join}\\
{\sl Input}: \= two collections ${\cal C}_{e_l}$  and  ${\cal C}_{e_r}$  of $(d,k,q)$-configurations of  $(X_{e_{l}},{\bf A}_{e_{l}})$ and  $(X_{e_{r}},{\bf A}_{e_{r}})$.\\
{\sl Output}: a collection ${\cal C}_r$  of $(d,k,q)$-configurations of  $(X_{e},{\bf A}_{e})$\\
(1) \= set ${\cal C}_e=\emptyset$\\
(2) \> for \= every pair $(Q_{e_l},Q_{e_r})\in {\cal C}_{e_l}\times {\cal C}_{e_r}$, if ${\bf merge}(Q_{e_l},Q_{e_r})\ne{\sf void}$,\\
 \> \>   then let ${\cal C}_e\leftarrow {\cal C}_e\cup\{{\bf merge}(Q_{e_l},Q_{e_r})\}$.\\
(3) return ${\cal C}_{e}$
\end{tabbing}
\end{minipage}
}
\end{center}

It remains to describe the routine {\bf merge}. For this, assume that  it receives as inputs the $(d,k,q)$-configurations 
$Q_{l}=(({\cal X}_{l},\chi_{l}),({\cal A}_{l},\alpha_{l}),({\cal F}_{l},{\cal E}_{l}),\delta_{l},z_{l})$ 
and  $Q_{r}=(({\cal X}_{r},\chi_{r}),({\cal A}_{r},\alpha_{r}),$
$({\cal F}_{r},{\cal E}_{r}),\delta_{r},z_{r})$ of $(X_{e_{l}},{\bf A}_{e_{l}})$ and $(X_{e_{r}},{\bf A}_{e_{r}})$ respectively. Procedure 
${\bf merge}(Q_{e_l},Q_{e_r})$ returns a $(d,k,q)$-configuration 
$(({\cal X},\chi),({\cal A},\alpha),({\cal F},{\cal E}),\delta,z)$ of $(X_{e},{\bf A}_{e})$ constructed as follows:

\begin{itemize}
\item[{\bf 1}.] If $z_{r}+z_{r}>q$, then return  {\sf void}, otherwise $z=z_{l}+z_{r}$

\noindent (This controls the number of 1-edges that are now contained in $\Delta_{e}$)

\item[{\bf 2}.]  Let $({\cal X}',\chi')=({\cal X}_{l},\chi_{l})\oplus ({\cal X}_{r},\chi_{r})$ and  if $\chi'^{-1}(2)\neq\emptyset$ then return {\sf void}.

\noindent (This  compute the ``fusion'' of the connected components of $(V(H_{e_{l}},E^{0}_{e_{l}}))$ and $(V(H_{e_{r}},E^{0}_{e_{r}}))$ with vertices in $V_{e_{l}}$ and $V_{e_{r}}$ and makes sure that none of the created components contains 2 or more $0$-vertices.) 

\item[{\bf 3}.] Let $({\cal X},\chi)=({\cal X}_{l}',\chi_{l}')|_{V_{e}}$

\noindent (This  computes the fusion $({\cal X}_{l}',\chi_{l}')$ is  restricted on the boundary $O_{e}$ of $\Delta_{e}$.)

\item[{\bf 4}.]  Let $({\cal A}',\alpha')=({\cal A}_{l},\alpha_{l})\oplus ({\cal A}_{r},\alpha_{r})$ and  if $\alpha'^{-1}(k+1)\neq\emptyset$ then return {\sf void}.

\item[{\bf 5}.] Let $({\cal A},\alpha)=({\cal A}_{l},\alpha_{l})\oplus ({\cal A}_{r},\alpha_{r})|_{{\bf A}_{e}}$.

\item[{\bf 6}.] Compute the function   $\gamma: ({\cal F}_{e_{l}}\cup {\cal F}_{e_{r}}\cup X_{e})\times ({\cal F}_{e_{l}}\cup {\cal F}_{e_{r}}\cup X_{e})\rightarrow \{0,\ldots,d+1\}$, whose description is given latter.

\item[{\bf 7}.]  Take  the disjoint union of  the graphs $({\cal F}_{l},{\cal E}_{l})$ and $({\cal F}_{r},{\cal E}_{r})$ 
 and  remove  from it every edge $\{\phi_{1},\phi_{2}\}$ for which $\gamma(\phi_{1},\phi_{2})\leq d$. Let ${\cal G}^{+}=({\cal F}^{+},{\cal E}^{+})$ be  the obtained  graph. 
 
\item[{\bf 8}.] If  for some edge $\{\phi_{1},\phi_{2}\}\in {\cal E}^+$ it holds that 
 for every $x_{1},x_{2}\in V_{e}$, $\gamma(\phi_{1},x_{1})+\gamma(\phi_{2},x_{2})>d$, then return {\sf void}. 
 
\item[{\bf 9}.] Consider  the function $\lambda: {\cal F}_{l}\cup {\cal F}_{r}\rightarrow \{1,\ldots,d\}^{X_{e}}$ such that
$\lambda(\phi)    =   \{(x,\gamma(\phi,x))\mid x\in X_{e}\}.$

\item[{\bf 10}.] For every $\phi'\in \lambda({\cal F}_{l}\cup {\cal F}_{r})$, do the following for every set ${\sf F}=\lambda^{-1}(\phi')$: identify in ${\cal G}^{+}$ all vertices in ${\sf F}$ and 
if at least one pair of them is adjacent in ${\cal G}^{+}$, then add an loop on the 
vertex created after this identification.
Let  ${\cal G}=({\cal F},{\cal E})$ be the resulting graph (notice that ${\cal F}=\lambda({\cal F}_{l}\cup {\cal F}_{r})$).

\item[{\bf 11}.]  $\delta=\{((x,x'),\gamma(x,x'))\mid x,x'\in V_{e}\}$.
\end{itemize}  

\medskip\noindent{\bf The definition of  function $\gamma$.}
We present here the definition of the function $\gamma$ used in the above description of the tables of the dynamic programming procedure.

Given a non-empty set $X$ and $q\in\{0,1\}$ we define 
\begin{eqnarray*}
{\sf ord}^{q}(X) & = & \{\pi\mid
\exists X'\subseteq X : X'\ne\emptyset \wedge\ |X'|\!\!\!\!\mod 2 = q\\
& & \ \ \ \ \ \ \ \ \ \ \ \ \ \ \ \ \ \ \ \ \ \wedge\ \pi \mbox{\ is an ordering of $X'$}\}
\end{eqnarray*}

Given $\gamma_{l}$ and $\gamma_{r}$, we define $\gamma: ({\cal F}_{e_{l}}\cup {\cal F}_{e_{r}}\cup X_{e})\times ({\cal F}_{e_{l}}\cup {\cal F}_{e_{r}}\cup X_{e})\rightarrow \{0,\ldots,d+1\}$ by distinguishing the following cases:

\begin{enumerate}
\item If  $(x\in X_{e}\setminus X_{e_r} \wedge \phi\in {\cal F}_{e_{l}})$ or $(x\in X_{e}\setminus X_{e_l} \wedge \phi\in {\cal F}_{e_{r}})$,  then
\begin{eqnarray*}
\gamma(\phi,x) & = & \min\big\{\phi(x),\min\{\phi(p_{1})+\sum_{\intv{1}{\rho-1}}\delta_{{\bf s}(i)}(p_{i},p_{i+1})+\\
& &\phantom{ \min\{} \delta_{{\bf s}(\rho)}(p_{\rho},x)\mid {[p_{1},\ldots,p_{\rho}]\in {\sf ord}^{0}(X_{e}^{F})}\}\big\},
\end{eqnarray*}
where
 ${\bf s}(i)=\mbox{\rm ``l''}$ if $\langle x\in X_{e}\setminus X_{e_{l}}\rangle=(i\!\!\! \mod 2)$, otherwise ${\bf s}(i)=\mbox{\rm ``r''}$.
\item If  $(x\in X_{e}\setminus X_{e_l} \wedge \phi\in {\cal F}_{e_{l}})$ or $(x\in X_{e}\setminus X_{e_r} \wedge \phi\in {\cal F}_{e_{r}})$,  then
\begin{eqnarray*}
\gamma(\phi,x) & = &  \min\big\{\phi(p_{1})+ \sum_{\intv{1}{\rho-1}}\delta_{{\bf t}(i)}(p_{i},p_{i+1})+\delta_{{\bf t}(\rho)}(p_{\rho},x)\\
& &\ \ \ \ \  \ \ \ \ \ \ \ \ \ \  \ \ \ \ \  \mid {[p_{1},\ldots,p_{\rho}]\in {\sf ord}^{1}(X_{e}^{F})}\}\big\},
\end{eqnarray*}
where
 ${\bf t}(i)=\mbox{\rm ``l''}$ if $\langle x\in X_{e}\setminus X_{e_{l}}\rangle\neq (i\!\!\! \mod 2)$, otherwise ${\bf t}(i)=\mbox{\rm ``r''}$.
 
\item If $x$ is one of the (at most two) vertices in $(X_{e_{r}}\cap X_{e_{r}})\setminus X_{e}^{F}$ and $\phi\in {\cal F}_{e_{l}}\cup {\cal F}_{e_{r}}$, then 
\begin{eqnarray*}
\gamma(\phi,x) & = & \min\big\{\phi(x),\\
& & \phantom{\min\{}\min\{\phi(p_{1})+ \sum_{\intv{1}{\rho-1}}\delta_{{\bf u}(i)}(p_{i},p_{i+1})+  \delta_{{{\bf u}(q)}}(p_{\rho},x)\\
& &\ \ \ \ \  \ \ \ \ \  \ \ \ \  \ \ \  \ \mid {[p_{1},\ldots,p_{\rho}]\in {\sf ord}^{q}(X_{e}^{F})}\}\mid q\in\{0,1\}\big\}
\end{eqnarray*}
where
 ${\bf u}(i)=\mbox{\rm ``r''}$ if $\langle \phi\in {\cal F}_{e_{l}}\rangle= (i\!\!\! \mod 2)$, otherwise ${\bf u}(i)=\mbox{\rm ``l''}$.

\item If $\phi,\phi'\in  {\cal F}_{l}\cup {\cal F}_{r}$, then
\begin{eqnarray*}
\gamma(\phi,\phi') & = & \min\big\{\phi(p_{1})+ \sum_{\intv{1}{\rho-1}}\delta_{{\bf u}(i)}(p_{i},p_{i+1}) + \phi'(p_{\rho})\\ & & \mid {[p_{1},\ldots,p_{\rho}]\in {\sf ord}^{q}(X_{e}^{F})}\big\}
\end{eqnarray*}
In this equality, $q=1$ if $\phi$ and $\phi'$ belong in different sets in $\{{\cal F}_{l},{\cal F}_{r}\}$, otherwise $q=0$. The function ${\bf u}$ is the same as in the previous case.

\item If $x_{1},x_{2}\in X_{e}\setminus X_{e_{r}}$ or  $x_{1},x_{2}\in X_{e}\setminus X_{e_{l}}$, then
\begin{eqnarray*}
\delta(x_{1},x_{2}) & =  & \min\big\{\delta_{{\bf y}(0,x_{1})}(x_{1},x_{2}), \min\{\delta_{{\bf y}(0,x_{1})}(x_{1},p_{1})+\\
& &    \sum_{i\in\intv{1}{\rho-1}}\delta_{{\bf y}(i,x_{1})}(p_{i},p_{i+1})+\\
& & \ \ \  \ \   \delta_{{\bf y}(0,x_{2})}(p_{\rho},x_{2})\mid  [p_{1},\ldots,p_{\rho}]\in {\sf ord}^{0}(X_{e}^{F}) \} \big\}
\end{eqnarray*}

In this equality   ${\bf y}(i,x)=\mbox{``l''}$ if $\langle x\in X_{e}\setminus X_{e_{r}}\rangle= \langle i\!\!\mod 2=0\rangle$ otherwise ${\bf y}(i,x)=\mbox{``r''}$.

\item If $x_{1},x_{2}$ belong in different sets is $\{X_{e}\setminus X_{e_{r}},X_{e}\setminus X_{e_{l}}\}$, then
\begin{eqnarray*}
\delta(x_{1},x_{2}) & = & \min\big\{\delta_{{\bf y}(0,x_{1})}(x_{1},p_{1}) +\sum_{\intv{1}{\rho-1}}\delta_{{\bf y}(i,x_{1})}(p_{i},p_{i+1})+\\
& & \phantom{\min\{}  \delta_{{\bf y}(0,x_{2})}(p_{\rho},x_{2})\mid  [p_{1},\ldots,p_{\rho}]\in {\sf ord}^{1}(X_{e}^{F}) \big\}
\end{eqnarray*}
The function ${\bf y}$ is the same as in the previous case.

\item If exactly one, say $x_{2}$, of $x_{1},x_{2}$ belongs in $X_{e_{r}}\cap X_{e_{r}})\setminus X_{e}^{F}$, then
\begin{eqnarray*}
\delta(x_{1},x_{2}) & = &  \min\bigg\{\delta_{{\bf y}(0,x_{1})}(x_{1},x_{2}),\\
& &  \   \min \big\{\min\{\delta_{{\bf y}(0,x_{1})}(x_{1},p_{1})+ \sum_{\intv{1}{\rho-1}}\delta_{{\bf y}(i,x_{1})}(p_{i},p_{i+1})+\\
& &  \    \delta_{{\bf y}(0,x_{2})}(p_{\rho},x_{2})\!\mid\!  [p_{1},\ldots,p_{\rho}]\in {\sf ord}^{q}(X_{e}^{F})\}\!\mid\!  q\in\{0,1\}\big\}\bigg\}
\end{eqnarray*}

The function ${\bf y}$ is the same as in the two previous cases.
In case $x_{1}$ belongs in  $X_{e_{r}}\cap X_{e_{r}})\setminus X_{e}^{F}$, then just swap the positions of $x_{1}$ and $x_{2}$ in the above equation.

\item If both $x_{1},x_{2}$ belong in $X_{e_{r}}\cap X_{e_{r}})\setminus X_{e}^{F}$, then
\begin{eqnarray*}
\delta(x_{1},x_{2})& = & \min\big\{\delta_{l}(x_{1},x_{2}),\delta_{r}(x_{1},x_{2}),\\
& & \phantom{\min\{}  \min \{\min\{\delta_{{\bf z}(0,j)}(x_{1},p_{1})+\\
& & \sum_{\intv{1}{\rho-1}}\delta_{{\bf z}(i,j)}(p_{i},p_{i+1})+ \delta_{{\bf z}(q,j)}(p_{\rho},x_{2})\!\mid\\
& & \ \ \ \ \  \ \ \ \    [p_{1},\ldots,p_{\rho}]\in {\sf ord}^{q}(X_{e}^{F})\}\!\mid\!  (q,j)\in\{0,1\}^2\}\big\}
\end{eqnarray*}
In the previous  equality, ${\bf z}(i,j)=\mbox{``l''}$ if $(i+j\!\!\! \mod 2) =0$, otehrwise ${\bf z}(i,x)=\mbox{``r''}$.

\end{enumerate}

\medskip\noindent{\bf Running time analysis.}
It now remains to prove that procedure {\bf join} runs in $(\alpha(q))^{2}\cdot 2^{O(k^2)+2^{O(b\log d)}}$ steps.
Recall that there exists a function $f$ such that
$|\mathfrak{T}(e)|\leq f(k,q,b,d)$. Therefore {\bf merge} will be called in Step (2) at most $(f(k,q,b,d))^2$ times. The first computationally non-trivial step of {\bf merge}  is Step 5, where 
function $\gamma$ is computed. Notice that $\gamma$ has at most $((d+1)^{|X_{e_{l}}|}+(d+1)^{|X_{e_{r}}|}+|X_{e}|)^{2}=2^{O(b\cdot \log d)}$ entries
and each of their values  require running over all permutations of the 
subsets of $X_{e}^{F}$ that are at most $b!=2^{O(b\cdot \log b)}$. These  
facts  imply that the computation of $\gamma$ takes $2^{O(b\cdot \log b)}$ steps. As Steps
6–10 deal with graphs of $2^{O(b\cdot \log d)}$ vertices, the running time of  {\bf join} is the claimed one.
\end{proof}

 We are now in position to prove the main algorithmic result of this paper.

\begin{proof}[Proof of Theorem~\ref{one}]
 Given an input $I=(G,q,k,d)$ of {\sc BBFPDC}, we 
 consider the 
 graph $H=G^{(\max\{3,k\})}$ whose construction takes $O(k^2n)$
 steps, because of Lemma~\ref{lowke}.
 Then run  the algorithm of Proposition~\ref{makesphere} with $(H,w)$ as 
 input, where $w=c_{1}\cdot c_{2}\cdot k\cdot d$.
 If the answer is that $\bw(H)>w$, then, From Proposition~\ref{lowke},
 $\tw(G)>c_{1}\cdot d$, therefore, from Lemma~\ref{btw}, we can safely report that 
 $I$ is a {\sc no}-instance.
 If the  algorithm of Proposition~\ref{makesphere} outputs a sphere-cut decomposition 
 $D=(T,\mu)$ of width at most $w=O(k\cdot d)$ then 
 we call the dynamic programming algorithm of Lemma~\ref{algdp}, with 
 input $(G,H,q,k,d,D,b)$. This, from Lemma~\ref{equiv}, provides an answer to {\sc BBFPDC}
 for the instance $I$ in $(\alpha(q))^{2}\cdot 2^{O((kd)^2\log d)+2^{O((kd)\log d)}} \cdot n=(\alpha(q))^{2}\cdot 2^{2^{O((kd)\log d)}} \cdot n$ steps and this completes the proof of the theorem.
\end{proof}

 \section{{\sf NP}-hardness proofs}
 \label{npouiwe3d}

In this section we show that the \textsc{Bounded Budget Plane Diameter Completion} and {\sc Bounded Budget/Face  Plane Diameter Completion} problems are {\sf NP}-complete. 

Here we consider $\Bbb{R}^{2}$-plane graphs, i.e., graphs embedded in the plane $\Bbb{R}^{2}$. 
Each $\Bbb{R}^{2}$-plane graph has exactly one unbounded face, called the \emph{outer} face, and all 
other faces are called \emph{inner faces}. 
Take in mind that every $\Bbb{S}_0$-plane graph has as many embeddings in $\Bbb{R}^{2}$ as the number of its faces (each correspond on which face of the embedding in $\Bbb{S}_{0}$ will be chosen to be the outer face in $\Bbb{R}^{2}$). All our problems can be equivalently restated on $\Bbb{R}^{2}$-plane graphs. We choose such embeddings because they facilitate the presentation of the result of this section.

We also need some additional terminology.
A  \emph{walk} in a graph $G$  of is a sequence $P=v_0,e_1,v_1,e_2,\ldots,e_s,v_s$ of vertices and edges of $G$ such that $v_0,\ldots,v_s\in V(G)$, $e_1,\ldots,e_s\in E(G)$,
the edges $e_1,\ldots,e_s$ are pairwise distinct, and 
for $i\in\{1,\ldots,s\}$, $e_i=\{v_{i-1},v_i\}$; $v_0,v_s$ are the \emph{end-vertices} of the walk. A walk is \emph{closed} if its end-vertices are the same.
The \emph{length} of a walk $P$ is the number of edges in $P$. 
For a walk $P$ with end-vertices $u,v$, we say that $P$ is a $(u,v)$-walk.
A walk is a \emph{path} if $v_0,\ldots,v_s$ and $e_1,\ldots,e_s$ are pairwise distinct with possible exception $v_0=v_s$, and a \emph{cycle} is a closed path.
We write $P=v_0\ldots v_s$ to denote a walk $P=v_0,e_1,\ldots,e_s,v_s$ omitting edges.

Recall that the \textsc{$3$-Satisfiability} problem for a given Boolean formula $\phi=C_1\wedge\ldots\wedge C_m$ with clauses $C_1,\ldots,C_m$ with 3 literals each over variables $x_1,\ldots,x_n$, asks whether $x_1,\ldots,x_n$ have an assignment that satisfies $\phi$. 
We write that a literal $x_i\in C_j$ ($\overline{x}_i\in C_j$ resp.) if this interval is in $C_j$.
For an instance $\phi$ of  \textsc{$3$-Satisfiability}, we define the graphs $G_\phi$ and $G_\phi'$ as follows. The vertex set of $G_\phi$ is $\{x_1,\ldots,x_n\}\cup\{C_1,\ldots,C_m\}$, and for $i\in\{1,\ldots,n\}$ and $j\in \{1,\ldots,m\}$, $\{x_i,C_j\}\in E(G_\phi)$ if and only if $C_j$ contains either $x_i$ or $\overline{x}_i$.  Respectively, $V(G_\phi')=\{x_1,\overline{x}_1,\ldots,x_n,\overline{x}_n\}\cup\{C_1,\ldots,C_m\}$ and $E(G_\phi')=\{\{x_i,\overline{x}_i\}|1\leq i\leq n\}\cup\{\{x_i,C_j\}|x_i\in C_j, 1\leq i\leq n,1\leq j\leq m\}\cup
\{\{\overline{x}_i,C_j\}|\overline{x}_i\in C_j, 1\leq i\leq n,1\leq j\leq m\}$.

Let $\phi$ over variables $x_1,\ldots,x_n$  be an instance of \textsc{$3$-Satisfiability} such that $G_\phi'$ is planar, and let $G'$ be a plane embedding of $G_\phi'$. 
Let also $R_\phi=\{\{x_i,\overline{x}_i\}|1\leq i\leq n\}\subseteq E(G')$. We define the bipartite graph $H(G')$ as the graph with the vertex set $R_\phi\cup F(G')$ and the edge set
$\{\{e,f\}|e\in R_\phi,f\in F(G')\text{ such that }e\text{ is incident to }f\}$.  

We consider the following special variant of \textsc{Satisfiability}.

\begin{center}
\fbox{\small\begin{minipage}{12.7cm}
\noindent{\sc Plane Satisfiability with Connectivity of Variables}\\
{\sl Input}: A Boolean formula $\phi=C_1\wedge\ldots\wedge C_m$ with clauses  $C_1,\ldots,C_m$ with at most 3 literals each over  variables
$x_1,\ldots,x_n$ such that $G'_\phi$ is planar, and a plane embedding 
 $G'$ of $G'_\phi$ such that $H(G')$ is connected.\\
{\sl Output}:  Is it possible to satisfy $\phi$?
\end{minipage}}
\end{center}

 We show that this problem is hard.

\begin{lemma}\label{lem:sat}
\textsc{Plane Satisfiability with Connectivity of Variables} is {\sf NP}-complete.
\end{lemma}

\begin{proof}
It is straightforward to see that \textsc{Plane Satisfiability with Connectivity of Variables} is in {\sf NP}. To show {\sf NP}-hardness, we reduce \textsc{Planar $3$-Satisfiability}, i.e. the \textsc{$3$-Satisfiability} problem restricted to instances $\phi$ such that $G_\phi$ is planar. This problem was shown to be {\sf NP}-complete by Lichtenstein in~\cite{Lichtenstein82}.

Let  $\phi=C_1\wedge\ldots\wedge C_m$ over variables $x_1,\ldots,x_n$ be an instance of \textsc{Planar $3$-Satisfiability}. For the plane graph $G_\phi$, we construct its plane embedding $G$.
It is well known that it can be done in polynomial time, e.g.,  by the classical algorithm of Hopcroft and Tarjan~\cite{HopcroftT74} or by the algorithm of Boyer and Myrvold~\cite{BoyerM04}.
We consequently consider variables $x_1,\ldots,x_n$ and 
modify $\phi$ and $G$.

Suppose that a variable $x_i$ occurs in the clauses $C_{j_1},\ldots,C_{j_{p(i)}}$. Without loss of generality we assume that the edges $\{x_i,C_{j_1}\},\ldots, \{x_i,C_{j_{p(i)}}\}$ are ordered clockwise in $G$ as shown in Fig.~\ref{fig:modif} a). We perform the following modifications of $\phi$ and $G$. 
\begin{itemize} 
\item Replace $x_i$ by $2p(i)$ new variables $x_{i,1},\ldots,x_{i,2p(i)}$.
\item For $k\in\{1,\ldots,p(i)\}$, replace $x_i$ in $C_{j_k}$ by $x_{i,2k-1}$.
\item Construct $2p(i)$ clauses $C_i^1,\ldots,C_i^{2p(i)}$ where $C_i^k=\overline{x}_{i,k-1}\vee x_{i,k}$ for $k\in\{1,\ldots2p(i)\}$; we assume that $x_{i,0}=x_{i,2p(i)}$.
\item Modify the current plane graph as it is shown in Fig.~\ref{fig:modif}.\footnote{Here and further we demonstrate constructions of plane embeddings in figures instead of long technical formal descriptions.}
\end{itemize}
Denote the obtained Boolean formula and plane graph by $\hat{\phi}$ and $\hat{G}$ respectively. By the construction, $\hat{G}$ is a plane embedding of $G_{\hat{\phi}}$.

\begin{figure}[ht]
\centering\scalebox{0.65}{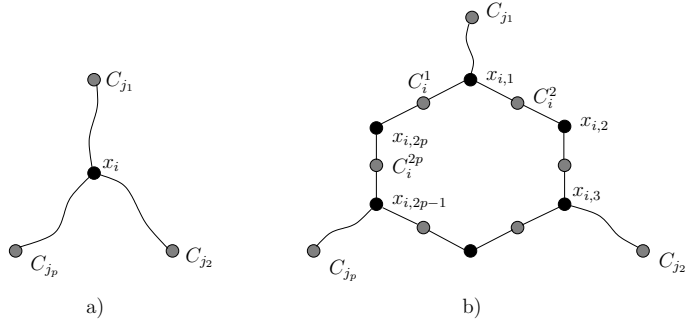}
\caption{Modification of $\phi$ and $G$: a) before the modification and b) after; $p=p(i)$.
\label{fig:modif}}
\end{figure}

We show that $\phi$ can be satisfied if and only if $\hat{\phi}$ has a satisfying assignment. Suppose that the variables have assigned values such that $\phi=true$. For each $i\in\{1,\ldots,n\}$, we assign the same value as $x_i$ for all the variables $x_{i,1},\ldots,x_{i,2p(i)}$ that replace $x_i$ in $\hat{\phi}$. It is straightforward to verify that $\hat{\phi}=true$ for this assignment. Assume now that $\hat{\phi}=true$ for some values of the variables. Observe that for each $i\in\{1,\ldots,n\}$, the variables $x_{i,1},\ldots,x_{i,2p(i)}$ that replace $x_i$ should have the same value to satisfy 
$C_i^1,\ldots,C_i^{2p(i)}$. It remains to observe that if each $x_i$ has the same value as $x_{i,1},\ldots,x_{i,2p(i)}$, then $\phi=true$ by the construction of $\hat{\phi}$.

\begin{figure}[ht]
\centering\scalebox{0.65}{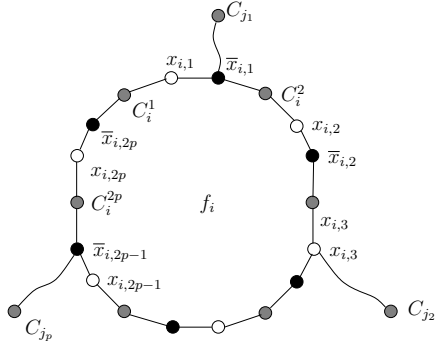}
\caption{Construction of $\hat{G}'$; it is assumed that $C_{j_1}$ contains $\overline{x}_i$, $C_{j_2}$ contains $x_i$ and $C_{j_{p(i)}}$ contains $\overline{x}_i$, and $p=p(i)$.
\label{fig:modif-hat}}
\end{figure}

Observe that each variable $x_{i,k}$ in $\hat{\phi}$ occurs in at most 3 clauses, and it occurs at least once in positive and at least once with negation. It implies that a plane embedding $\hat{G}'$ of $G_{\hat{\phi}}'$ can be constructed from $\hat{G}$ by ``splitting'' the variable vertices as shown in Fig.~\ref{fig:modif-hat}.  Clearly,   $\hat{G}'$ can be constructed in polynomial time.

We claim that $H(\hat{G}')$ is connected. To see it, observe that $\hat{G}'$ is constructed from $G$ by replacing each variable-vertex $x_i$ by the cycle  $L_i=C_i^1x_{i,1}\overline{x}_{i,1}C_i^2\ldots C_i^1$ (see Fig.~\ref{fig:modif-hat}). Respectively, this graph has $n$ new faces that are inner faces of these cycles. All other faces correspond to the faces of $G$. Denote by $f_i$ the inner face of $L_i$ for $i\in\{1,\ldots,n\}$. 
 Notice that $R_{\hat{\phi}}$ contains edges from the cycles $L_i$. It follows that each vertex of $R_{\hat{\phi}}$ is adjacent to some vertex $f_i$ in  $H(\hat{G}')$. Hence, to prove the connectivity of  $H(\hat{G}')$, it is sufficient to show that for any two vertices $h_1,h_2\in F(\hat{G}' )$,  $H(\hat{G}')$ has a $(h_1,h_2)$-walk.

Consider the dual $G^*$ of $\hat{G}'$. Recall that $V(G^*)=F(\hat{G}')$ and two vertices of $G^*$ are adjacent if and only if the corresponding faces of $\hat{G}'$ are adjacent. It is straightforward to observe that the dual of any plane graph is always connected. Hence, to show that for any two vertices $h_1,h_2\in F(\hat{G}' )$ of $H(\hat{G}')$,  $H(\hat{G}')$ has a $(h_1,h_2)$-walk, it is sufficient to prove that it holds for any two $h_1,h_2$ that are adjacent vertices of $G^*$, i.e., adjacent faces of $\hat{G}'$. 
Suppose that $h_1=f_i$ for some $i\in\{1,\ldots,n\}$. Then $h_2$ is a face corresponding to a face $h_2'$ of $G$ such that the vertex $x_i$ lies on the boundary of $h_2'$. Then by the construction of  $\hat{G}'$, there is an edge $e=\{x_{i,j},\overline{x}_{i,j}\}$  of $\hat{G}'$ that lies on the boundaries of $h_1$ and $h_2$. Because $e$ is a vertex of  $H(\hat{G}')$ adjacent to $h_1,h_2$,  there is a $(h_1,h_2)$-walk in $H(\hat{G}')$.
Assume now that $h_1,h_2$ are faces of $\hat{G}'$ distinct from $f_i$ for $i\in\{1,\ldots,n\}$. Because $h_1,h_2$ are adjacent in $G^*$, the faces $h_1,h_2$ correspond to faces $h_1',h_2'$ of $G$ such that $h_1',h_2'$ has a common vertex $x_i$ on their boundaries. It implies that $h_1,h_2$ are adjacent to $f_i$ in $G^*$. We already proved that $H(\hat{G}')$ has $(f_i,h_1)$ and $(f_i,h_2)$-walks. Therefore, $H(\hat{G}')$ has an $(h_1,h_2)$-walk. 

It completes the proof of connectedness of  $H(\hat{G}')$ and the proof of the lemma.
\end{proof}

For the proof of our main result, we need some special gadgets. We introduce them and prove their properties that will be useful further.  

Let $r\geq 3$ be a positive integer. We construct the graph $W_r(v_1,\ldots,v_r)$ as follows (see Fig.~\ref{fig:web}). 
\begin{itemize}
\item Construct vertices $v_1,\ldots,v_r$ and a vertex $u$.
\item For $i\in\{1,\ldots,r\}$, construct a $(v_i,u)$ path $x_0^i\ldots x_r^i$  of length $r$, $v_i=x_0^i$, $u=x_r^i$.
\item For $j\in\{1,\ldots,r-1\}$, construct a cycle $x_j^1\ldots x_j^rx_j^1$.
\item For $i\in\{1,\ldots,r\}$ and $j\in\{1,\ldots,r-1\}$, construct an edge $\{x_{j-1}^{i-1},x_j^i\}$; we assume that $x_j^0=x_j^r$ for $j\in \{0,\ldots,r\}$.
\end{itemize}
We say that the vertices of $V(W_r(v_1,\ldots,v_r))\setminus\{v_1,\ldots,v_r\}$ are the \emph{inner} vertices of the gadget.

\begin{figure}[ht]
\centering\scalebox{0.65}{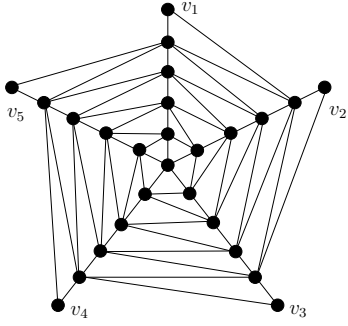}
\caption{Construction of $W_5(v_1,\ldots,v_5)$.
\label{fig:web}}
\end{figure}

Let $G$ be a plane graph with a face $f$, and let $v_1\ldots v_rv_1$, $r\geq 3$, be a facial walk for $f$. We say that $G'$ is obtained from $G$ by \emph{attaching a web} to $f$ if $G'$ is constructed 
 by adding a copy of  $W_r(v_1,\ldots,v_r)$, where the vertices $v_1,\ldots,v_r$ of the gadget are identified with the vertices with the same names in the facial walk, and embedding  $W_r(v_1,\ldots,v_r)$ if $f$ as is shown in Fig.~\ref{fig:web-att}. Notice that some vertices in the facial walk can occur several times.   

\begin{figure}[ht]
\centering\scalebox{0.65}{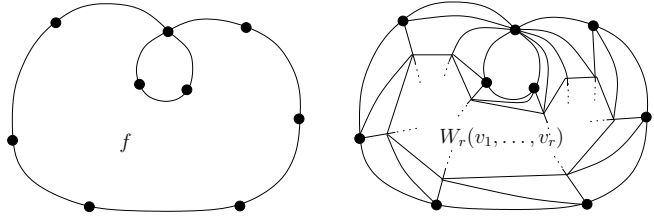}
\caption{Attachment of a web.
\label{fig:web-att}}
\end{figure}

\begin{lemma}\label{lem:web}
Let $G$ be a plane graph with a face $f$ that has a facial walk of length $r\geq 3$, and let $G'$ be a plane graph obtained from $G$ by attaching a web to $f$. 
\begin{itemize}
\item[i)] For any two vertices $u,v\in V(G)$, $\dist_{G'}(u,v)=\dist_G(u,v)$. Moreover, any shortest $(u,v)$-path in $G'$ has no inner vertices of $W_r(v_1,\ldots,v_r)$ attached to $f$. 
\item[ii)] For any vertex $v\in V(W_r(v_1,\ldots,v_r))$, there is a vertex $u\in V(G)$ such that $\dist_{G'}(u,v)\leq r$.
\end{itemize}
\end{lemma}

\begin{proof}
Let $v_1\ldots v_rv_0$ be a facial walk for $f$. To prove i), it is sufficient to observe that for all $v_i,v_j$, the length of any $(v_i,v_j)$-path in  $W_r(v_1,\ldots,v_r)$ is greater that the length of a shortest $(v_i,v_j)$-path in $G$ that lies on the boundary of $f$. The definition of  $W_r(v_1,\ldots,v_r)$ immediately implies ii).
\end{proof}

Let $h$ be a positive integer.
The graph $M_h(u_1,u_2,u_3)$ is defined as follows (see Fig.~\ref{fig:mast}). 
\begin{itemize}
\item Construct vertices $u_1,u_2,u_3$ and $v_1,v_2,v_3$.
\item For $i\in\{1,2,3\}$, construct a $(u_i,v_i)$ path $x_0^i\ldots x_r^i$  of length $\ell$, $u_i=x_0^i$, $v_i=x_h^i$.
\item For $j\in\{1,\ldots,h\}$, construct a cycle $x_j^1\ldots x_j^rx_j^1$.
\item For $j\in\{1,\ldots,h\}$, construct  edges $\{x_{j-1}^1,x_j^2\}$, $\{x_{j-1}^1,x_j^3\}$ and $\{x_{j-1}^2,x_j^3\}$.
\end{itemize}
We say that the vertices of $V(M_h(u_1,u_2,u_3))\setminus\{u_1,u_2,u_3\}$ are the \emph{inner} vertices of the gadget.
We also say that $u_1$ is the \emph{root} and $v_1$ is the \emph{pole} of $M_\ell(u_1,u_2,u_3)$.

\begin{figure}[ht]
\centering\scalebox{0.65}{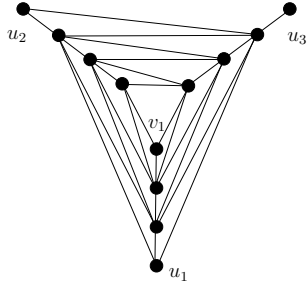}
\caption{Construction of $M_3(u_1,u_2,u_3)$.
\label{fig:mast}}
\end{figure}

Let $G$ be a plane graph, and let $u_1\in V(G)$ be a vertex incident to a face $f$ with a triangle facial walk $u_1u_2u_3u_1$. 
Let also $\ell$ be a positive integer.
We say that $G'$ is obtained from $G$ by \emph{attaching a mast of height $h$} rooted in $u_1$ to $f$ if $G'$ is constructed
 by adding a copy of  $M_h(u_1,u_2,u_3)$, where the vertices $u_1,u_2,u_3$ of the gadget are identified with the vertices with the same names in the facial walk, and embedding  $M_h(u_1,u_2,u_3)$ in $f$. We need the properties summarized in the following straightforward lemma.

\begin{lemma}\label{lem:mast}
Let $\ell$ be a positive integer. Let $G$ be a plane graph, and let $u_1$ be a vertex of $G$ incident to a face $f$ with a triangle facial walk $u_1u_2u_3u_1$.
Let also  $G'$ be a plane graph obtained from $G$ by attaching a mast of height $h$ rooted in $u_1$ to $f$. 
\begin{itemize}
\item[i)] For any two vertices $u,v\in V(G)$, $\dist_{G'}(u,v)=\dist_G(u,v)$. Moreover, any shortest $(u,v)$-path in $G'$ has no inner vertices of $M_h(u_1,u_2,u_2)$ attached to $f$. 
\item[ii)] For any vertex $v\in V(M_h(u_1,u_2,u_2))$,  $\dist_{G'}(u_1,v)\leq h$.
\item[iii)] If $v$ is the pole of $M_h(u_1,u_2,u_2)$, then $\dist_{G'}(u_1,v)=h$ and $\dist_{G'}(u_2,v)> h, \dist_{G'}(u_3,v)> h$.
\item[iv)] For any inner vertices $x,y$ of $M_h(u_1,u_2,u_2)$, $\dist_{G'}(x,y)\leq h$. 
\end{itemize}
\end{lemma}

Now we are ready to prove the main result of the section.

\begin{proof}[Proof of Theorem~\ref{thm:NP-c}]
It is straightforward to see that {\sc BPDC} and {\sc BFPDC}  are  in {\sf NP}. To show {\sf NP}-hardness, we reduce \textsc{Plane Satisfiability with Connectivity of Variables} that was shown to be {\sf NP}-complete in Lemma~\ref{lem:sat}.

First, we consider {\sc BPDC}.

Let $(\phi,G')$ be an instance of \textsc{Plane Satisfiability with Connectivity of Variables}, where 
$\phi=C_1\wedge\ldots\wedge C_m$ is a Boolean formula with clauses $C_1,\ldots,C_m$ with at most 3 literals each over  variables $x_1,\ldots,x_n$ such that $G'_\phi$ is planar,
and $G'$ is a plane embedding of $G'_\phi$ such that $H(G')$ is connected. 
Recall that $H(G')$  is the bipartite graph with the bipartition of the vertex set $(R_\phi, F(G'))$, where $R_\phi=\{\{x_i,\overline{x}_i\}|1\leq i\leq n\}\subseteq E(G')$, and $F(G')$ is the set of faces of $G'$, and 
for $e\in R_\phi$ and $f\in F(G')$,
$\{e,f\}\in E(H(G'))$ if and only if the edge $e$ is incident to the face $f$ in $G'$. Notice that $\deg_{H(G')}(e)\leq 2$ for any $e\in R_\phi$.

We select an arbitrary vertex $r\in F(G')$ of $H(G')$. Using the connectedness of $H(G')$, we find in polynomial time a tree $T$ of shortest $(r,e)$-paths for $e\in R_\phi$ by the breadth-first search. We assume that $T$ is rooted in $r$ and it defines the parent-child relation on $T$. Let $L\subseteq R_\phi$ be the set of leaves of $T$, and let
$s=\max\{\dist_T(r,e)|e\in L\}$.

\begin{figure}[ht]
\centering\scalebox{0.65}{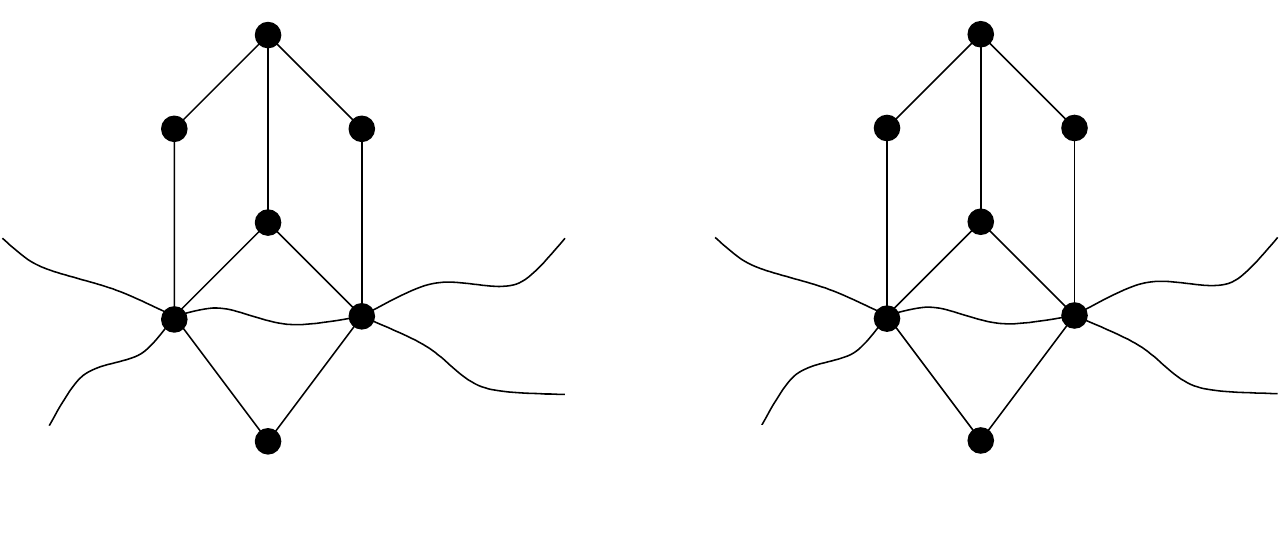}
\caption{Construction of gadgets for $\{x_i,\overline{x}_i\}$.
\label{fig:edge}}
\end{figure}

We construct the plane graph $\hat{G}$ as follows.
\begin{itemize}
\item[i)] Construct a copy of $G'$. 
\item[ii)] For each vertex $f\in V(T)$ such that $f\in F(G')$, crate a vertex $v_f$ embedded in the face $f$.  
\item[iii)] For each $e=\{x_i,\overline{x_i}\}\in R_\phi\setminus L$, denote by $p$ its parent and by $c$ its child in $T$, construct vertices $y_i,\overline{y}_i,z_i$ and edges 
$\{x_i,y_i\}, \{y_i,v_p\},\{x_i,v_c\}, \{x_i,z_i\}$, $\{\overline{x}_i,\overline{y}_i\}, \{\overline{y}_i,v_p\},\{\overline{x}_i,v_c\},  \{\overline{x}_i,z_i\},\{z_i,v_p\}$ and embed them as is shown in Fig.~\ref{fig:edge} a).
Denote by $f_i$ the inner face of the cycle $x_iy_iv_pz_ix_i$ and by $\overline{f}_i$  the inner face of the cycle $\overline{x}_i\overline{y}_iv_pz_i\overline{x}_i$.
\item[iv)] For each $e=\{x_i,\overline{x_i}\}\in L$, denote by $p$ its parent in $T$, construct vertices $y_i,\overline{y}_i,z_i,w_i$ and edges 
$\{x_i,y_i\}, \{y_i,v_p\},\{x_i,w_i\}, $ $\{x_i,z_i\}$, $\{\overline{x}_i,\overline{y}_i\},$ $ \{\overline{y}_i,v_p\},\{\overline{x}_i,w_i\}, $ $ \{\overline{x}_i,z_i\},\{z_i,v_p\}$ and embed them as is shown in Fig~\ref{fig:edge} b).
Denote by $f_i$ the inner face of the cycle $x_iy_iv_pz_ix_i$ and by $\overline{f}_i$  the inner face of the cycle $\overline{x}_i\overline{y}_iv_pz_i\overline{x}_i$.
\item[v)] For each $i\in\{1,\ldots,n\}$ and $j\in\{1,\ldots,m\}$, if $\{x_i,C_j\}\in E(G')$ ($\{\overline{x}_i,C_j\}\in E(G')$ resp.), replace this edge by a $(x_i,C_j)$-path (by $(\overline{x}_i,C_j)$-path resp.) of length $2s-\dist_{T}(r,\{x_i,\overline{x}_i\})$.
\end{itemize}
We denote the constructed at this stage graph by $\hat{G}_1$. Observe that $\hat{G}_1$ is connected. Hence, each face has a facial walk. Denote by $\ell$ the length of a longest facial walk in $\hat{G}_1$.  Now we proceed with construction of $\hat{G}$.
\begin{itemize}
\item[vi)] For each face $f\in F(\hat{G}_1)$ distinct from the faces $f_i,\overline{f}_i$ for $i\in\{1,\ldots,n\}$, attach a web to $f$.
\end{itemize}
Denote the constructed at this stage graph by $\hat{G}_2$. Notice that $\hat{G}_2$ is 3-connected due to attached webs.
\begin{itemize}
\item[vii)] For $j\in\{1,\ldots,m\}$, select a face $f$ of the obtained graph such that $C_j$ is incident to $f$ and attach a mast of height $\ell+2s$ rooted in $C_j$ to $f$ (notice that the boundary of $f$ is a triangle because of attached webs).
\item[viii)]  For each $e=\{x_i,\overline{x_i}\}\in L$, attach a mast of height $\ell+4s-1-\dist_T(r,e)$ rooted in $w_i$ to the face with the facial walk $w_ix_i\overline{x}_iw_i$.
\item[ix)] For the vertex $v_r$, select a face $f$ with a triangle boundary such that $v_r$ is incident to $f$ (such a face always exists due to attached webs) and attach a mast of height $\ell+8s$ rooted in $v_r$ to $f$.
\end{itemize}
Notice that the obtained graph $\hat{G}$ is 3-connected  because $\hat{G}_2$ is 3-connected and attachments of masts cannot destroy 3-connectivity. Also only the faces $f_i,\overline{f}_i$ for $i\in\{1,\ldots,n\}$ have degree 4, and all other faces have degree 3.

To complete the construction of an instance of {\sc BPDC}, we set $q=n$ and $d=2\ell+12s$.

We show that $(\phi,G')$ is a yes-instance of \textsc{Plane Satisfiability with Connectivity of Variables} if and only if $(\hat{G},q,d)$ is a yes-instance of  {\sc BPDC}.

Suppose that $(\phi,G')$ is a yes-instance of \textsc{Plane Satisfiability with Connectivity of Variables}. Assume that the variables $x_1,\ldots,x_n$ have values such that $\phi=true$.  For $i\in\{1,\ldots,n\}$, if $x_i=true$, then we add an edge $\{x_i,v_p\}$ for the parent $p$ of $\{x_i,\overline{x}_i\}$ in $T$ and embed this edge in $f_i$. Respectively, 
we add an edge $\{\overline{x}_i,v_p\}$  and embed this edge in $\overline{f}_i$ if $x_i=false$. Denote the obtained graph by $\hat{G}'$. We show that $\diam(\hat{G}')\leq d$.

By the construction of $\hat{G}_1$, for any vertex $v\in V(\hat{G}_1)$, $\dist_{\hat{G}_1}(v_r,v)\leq 3s$. 
By Lemma~\ref{lem:web}, any vertex  $v\in V(\hat{G}_2)$ is at distance at most $\ell$ from a vertex of $\hat{G}_1$ in $\hat{G}_2$. Hence, for any vertex $v\in V(\hat{G}_2)$, $\dist_{\hat{G}_2}(v_r,v)\leq \ell+3s$. Observe also that for any $e=\{x_i,\overline{x_i}\}\in L$, $\dist_{\hat{G}'}(v_r,w_i)=\dist_T(r,e)+1$.
To show that for any $u,v\in V(\hat{G}')$, $\dist_{\hat{G}'}(u,v)\leq d$, we consider
five cases.

\noindent
{\bf Case 1.} $u,v\in V(\hat{G}_2)$. Because $\dist_{\hat{G}_2}(v_r,u)\leq \ell+3s$ and $\dist_{\hat{G}_2}(v_r,v)\leq \ell+3s$, $\dist_{\hat{G}'}(u,v)\leq \dist_{\hat{G}_2}(u,v)\leq 2\ell+ 6s\leq d$.

\noindent
{\bf Case 2.} $u,v$ are vertices of the same mast attached to a face of $\hat{G}_2$. By Lemma~\ref{lem:mast}, $\dist_{\hat{G}'}(u,v)$ is at most the height of the mast, and we have that 
 $\dist_{\hat{G}'}(u,v)\leq \ell+8s\leq d$.

\noindent
{\bf Case 3.} $u\in V(\hat{G}_2)$ and $v$ is a vertex of a mast attached to a face of $\hat{G}_2$. By Lemma~\ref{lem:mast}, $\dist_{\hat{G}'}(u,v_r)\leq \ell+8s$ if the mast is rooted in $v_r$.
Suppose that this mast is rooted in some other vertex $z$, i.e., $z=w_i$ or $z=C_j$ for some $i\in\{1,\ldots,n\}$, $j\in \{1,\ldots,m\}$. Then
$\dist_{\hat{G}'}(u,v_r)\leq \ell+4s-1+\dist_{\hat{G}_1}(z,r)\leq \ell+8s$.
Because $\dist_{\hat{G}'}(v_r,v)\leq\dist_{\hat{G}_2}(v_r,v)\leq \ell+3s$, $\dist_{\hat{G}'}(u,v)\leq 2\ell+11s\leq d$.

\noindent
{\bf Case 4.} $u,v$ are vertices of distinct masts attached to faces of $\hat{G}_2$ that are rooted in $z,z'\neq v_r$ respectively. If $z=w_i$ for some $i\in\{1,\ldots,n\}$, then 
$\dist_{\hat{G}'}(u,v_r)\leq \ell+4s-1-\dist_T(r,e)+\dist_{\hat{G}'}(v_r,w_i)\leq (\ell+4s-1-\dist_T(r,e))+(\dist_T(r,e)+1)\leq \ell+4s$ where $e=\{x_i,\overline{x}_i\}$.
If $z=C_j$ for some $j\in\{1,\ldots,m\}$, then  $\dist_{\hat{G}'}(u,v_r)\leq \ell+2s+\dist{\hat{G}_1}(C_j,v_r)\leq \ell+5s$. Clearly, the same bounds hold for $\dist_{\hat{G}'}(v,v_r)$.
We have that  $\dist_{\hat{G}'}(u,v)\leq \dist_{\hat{G}'}(u,v_r)+\dist_{\hat{G}'}(v_r,v)\leq 2\ell+10s\leq d$.

It remains to consider the last case.

\noindent
{\bf Case 5.} $u,v$ are vertices of masts attached to faces of $\hat{G}_2$ such that $u$ is in the mast rooted in $v_r$ and $v$ is in a mast rooted in $z\neq v_r$. 
Suppose that $z=w_i$ for  some $i\in\{1,\ldots,n\}$. Then $e=\{x_i,\overline{x}_i\}\in L$. We have that 
 $\dist_{\hat{G}'}(u,v)\leq \dist_{\hat{G}'}(u,v_r)+\dist_{\hat{G}'}(v_r,w_i)+\dist_{\hat{G}'}(w_i,v)\leq
(\ell+8s)+(\dist_T(r,e)+1)+(\ell+4s-1-\dist_T(r,e))\leq 2\ell+12s\leq d$. Assume that $z=C_j$ for $j\in\{1,\ldots,m\}$. Then the clause $C_j$ in $\phi$ contains a literal that has the value $true$.
Let $x_i$ be such a literal (the case when $C_j$ contains some $\overline{x}_i=true$ is symmetric).  Notice that if $x_i=true$, then for the vertex $x_i\in V(\hat{G}')$, 
$\dist_{\hat{G}'}(x_i,v_r)=\dist_T(e,r)$ for $e=\{x_i,\overline{x}_i\}$ by the construction of $\hat{G}$ and the selection of the added edges. Then, 
 $\dist_{\hat{G}'}(u,v)\leq \dist_{\hat{G}'}(u,v_r)+\dist_{\hat{G}'}(v_r,x_i)+\dist_{\hat{G}'}(x_i,C_j)+\dist_{\hat{G}'}(C_j,v)\leq
(\ell+8s)+\dist_T(r,e)+(2s-\dist_T(r,e))+(\ell+2s)\leq 2\ell+12s\leq d$.

Suppose now that $(\hat{G},q,d)$ is a yes-instance of  {\sc BPDC}. Let $A$ be a set of at most $q$ edges such that the graph $\hat{G}'$ obtained from $\hat{G}$ by the addition of $A$ has diameter at most $d$. Because only the faces $f_i,\overline{f}_i$ for $i\in\{1,\ldots,n\}$ have degree 4 and all other faces have degree 3, each edge of $A$ has its end-vertices in the boundary of some $f_i$ or $\overline{f}_i$ and is embedded in this face. Using this observation, denote by $\hat{G}_1'$ and $\hat{G}_2'$ the graphs obtained from $\hat{G}_1$ and $\hat{G}_2$ respectively by the addition of $A$.  Let $v_r'$ be the pole of the mast rooted in $v_r$. Because $\diam(\hat{G}')\leq d$, for any $u\in V(\hat{G}')$, $\dist_{\hat{G}'}(v_r',u)\leq d$ and, in particular, it holds for poles of other masts.

Consider masts rooted in $w_i$ for $e=\{x_i,\overline{x}_i\}\in L$. For a mast rooted in $w_i$, denote by $w_i'$ its pole. By Lemma~\ref{lem:mast}, $\dist_{\hat{G}'}(v_r',w_i')=
\dist_{\hat{G}'}(v_r',v_r)+\dist_{\hat{G}'}(v_r,w_i)+\dist_{\hat{G}'}(w_i,w_i')=(\ell+8s)+ \dist_{\hat{G}'_2}(v_r,w_i)+(\ell+4s-1-\dist_T(r,e))$, and by Lemma~\ref{lem:web},
$\dist_{\hat{G}'_2}(v_r,w_i)=\dist_{\hat{G}'_1}(v_r,w_i)$. We conclude that $\dist_{\hat{G}'_1}(v_r,w_i)\leq \dist_T(r,e)+1$. Because $\dist_T(r,e)+1\leq s+1$, a shortest $(v_r,w_i)$-path in 
$\hat{G}_1'$ does not contain the vertices $C_j$ for $j\in\{1,\ldots,m\}$. We obtain that for every edge $e'=\{x_h,\overline{x}_h\}$ that lies on the unique $(r,e)$-path in $T$, $\{x_i,v_p\}\in A$ or $\{\overline{x}_i,v_p\}\in A$ where $p$ is the parent of $e'$ in $T$. This holds for each leaf of $T$.
Because $R_\phi\subseteq V(T)$ and $k=n$, we have that for each $h\in\{1,\ldots,n\}$, either
$\{x_i,v_p\}\in A$ or $\{\overline{x}_i,v_p\}\in A$ where $p$ is the parent of $\{x_h,\overline{x}_h\}$ in $T$. 
For $h\in \{1,\ldots,n\}$, we let the variable $x_h=true$ if $\{\overline{x}_i,v_p\}\in A$ and $x_h=false$ otherwise. 
We show that this assignment satisfies $\phi$. 

Consider a clause $C_j$ for $j\in\{1,\ldots,m\}$. To simplify notations, 
assume that $C_j$ contains literals $x_{i_1},x_{i_2},x_{i_3}$ (the cases when $C_j$ contains two literals and/or some literals are negations of variables are considered in the same way). Let $C_j'$ be the pole of the mast rooted in the vertex $C_j$. We have that $\dist_{\hat{G}'}(v_r',C_j')\leq d$. 
By Lemma~\ref{lem:mast},$\dist_{\hat{G}'}(v_r',C_j')=
\dist_{\hat{G}'}(v_r',v_r)+\dist_{\hat{G}'}(v_r,C_j)+\dist_{\hat{G}'}(C_j,C_j')=(\ell+8s)+ \dist_{\hat{G}'_2}(v_r,C_j)+(\ell+2s)$,
 and by Lemma~\ref{lem:web}, $\dist_{\hat{G}'_2}(v_r,C_j)=\dist_{\hat{G}'_1}(v_r,C_j)$.
Therefore, $\dist_{\hat{G}'_1}(v_r,C_j)\leq 2s$. Let $e_h=\{x_{i_h},\overline{x}_{i_h}\}$ for $h\in\{1,2,3\}$. By the construction of $\hat{G}'$, 
$\dist_{\hat{G}'_1}(v_r,C_j)=\min\{\dist_{\hat{G}'_1}(v_r,x_{i_h})+(2s-\dist_T(r,e_h))|1\leq h\leq 3 \}$.
Let $\dist_{\hat{G}'_1}(v_r,C_j)=\dist_{\hat{G}'_1}(v_r,x_{i_h})+(2s-\dist_T(r,e_h))$ for $h\in\{1,2,3\}$.
It follows that $\dist_{\hat{G}'_1}(v_r,x_{i_h})\leq \dist_T(r,e_h)$, and this immediately implies that $\{v_p,x_{i_h}\}\in A$ where $p$ is the parent of $e_h$ in $T$. By the definition, $x_{i_h}=true$ and, therefore, $C_j=true$. This holds for each $C_j$ for $j\in\{1,\ldots,m\}$, and we conclude that $\phi=true$.

To complete  the proof of the {\sf NP}-hardness of {\sc BPDC}, it remains to observe that $\hat{G}$ can be constructed in polynomial time.

To show {\sf NP}-hardness of {\sc BFPDC}, we use similar arguments. 

Let $(\phi,G')$ be an instance of \textsc{Plane Satisfiability with Connectivity of Variables}, where 
$\phi=C_1\wedge\ldots\wedge C_m$ is a Boolean formula with clauses $C_1,\ldots,C_m$ with at most 3 literals each over  variables $x_1,\ldots,x_n$ such that $G'_\phi$ is planar,
and $G'$ is a plane embedding of $G'_\phi$ such that $H(G')$ is connected. 
As before, we pick an arbitrary vertex $r\in F(G')$ of $H(G')$ and  find  a tree $T$ rooted in $r$ of shortest $(r,e)$-paths for $e\in R_\phi$ 
with the set of leaves  $L\subseteq R_\phi$. 
Let $s=\max\{\dist_T(r,e)|e\in L\}$.

\begin{figure}[ht]
\centering\scalebox{0.65}{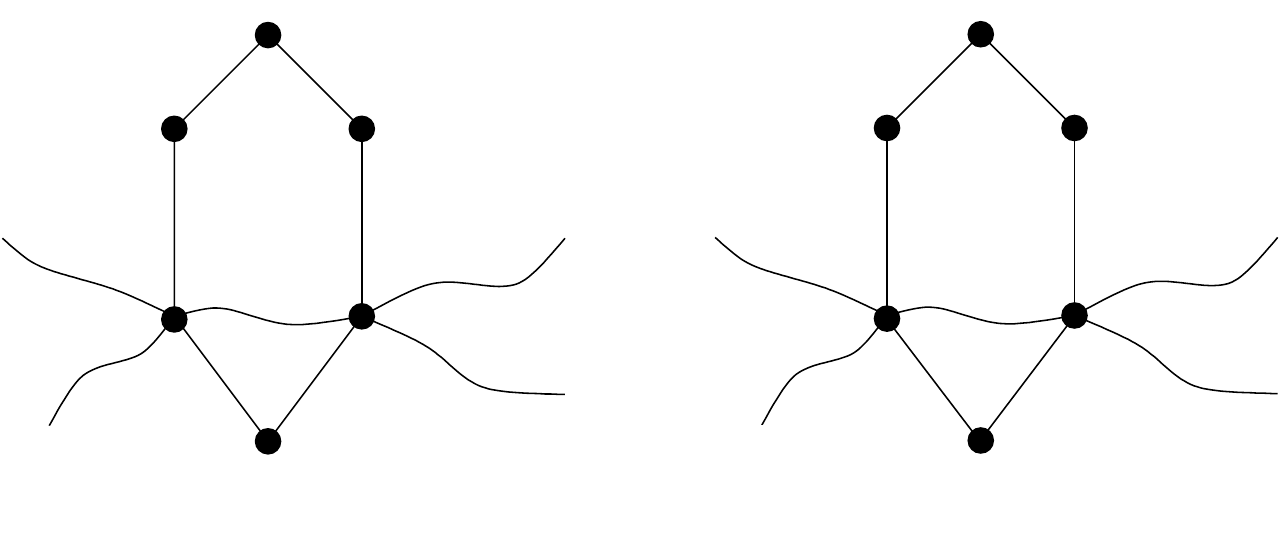}
\caption{Construction of gadgets for $\{x_i,\overline{x}_i\}$.
\label{fig:edge-ii}}
\end{figure}

We construct the plane graph $\tilde {G}$ similarly to the construction of $\hat{G}$ above. The only difference is that Steps iii) and iv) are replaced by the following steps iii$^*$) and iv$^*$).
\begin{itemize}
\item[iii$^*$)] For each $e=\{x_i,\overline{x_i}\}\in R_\phi\setminus L$, denote by $p$ its parent and by $c$ its child in $T$, construct vertices $y_i,\overline{y}_i$ and edges 
$\{x_i,y_i\}, \{y_i,v_p\},\{x_i,v_c\}$, $\{\overline{x}_i,\overline{y}_i\}, \{\overline{y}_i,v_p\},\{\overline{x}_i,v_c\}$ and embed them as is shown in Fig.~\ref{fig:edge-ii} a).
Denote by $f_i$ the inner face of the cycle $x_iy_iv_p\overline{y}_i\overline{x}_ix_i$.
\item[iv$^*$)] For each $e=\{x_i,\overline{x_i}\}\in L$, denote by $p$ its parent in $T$, construct vertices $y_i,\overline{y}_i,w_i$ and edges 
$\{x_i,y_i\}, \{y_i,v_p\},\{x_i,w_i\}$, $\{\overline{x}_i,\overline{y}_i\}, \{\overline{y}_i,v_p\},\{\overline{x}_i,w_i\}$ and embed them as is shown in Fig~\ref{fig:edge-ii} b).
Denote by $f_i$ the inner face of the cycle $x_iy_iv_p\overline{y}_i\overline{x}_ix_i$.
\end{itemize}
Observe that $\tilde{G}$ can be obtained from $\hat{G}$ by the deletion of the vertices $z_1,\ldots,z_n$, and for any $u,v\in V(\tilde{G})$, $\dist_{\tilde{G}}(u,v)=\dist_{\hat{G}}(u,v)$.
Notice that the obtained graph $\tilde{G}$ is 3-connected,  the faces $f_1,\ldots,f_n$ have degree 5, and all other faces have degree 3.
To complete the construction of an instance of {\sc BFPDC}, we set $k=1$ and $d=2\ell+12s$.

We show that $(\phi,G')$ is a yes-instance of \textsc{Plane Satisfiability with Connectivity of Variables} if and only if $(\tilde{G},k,d)$ is a yes-instance of  {\sc  BFPDC}.

Suppose that $(\phi,G')$ is a yes-instance of \textsc{Plane Satisfiability with Connectivity of Variables}. Assume that the variables $x_1,\ldots,x_n$ have values such that $\phi=true$.  For $i\in\{1,\ldots,n\}$, if $x_i=true$, then we add an edge $\{x_i,v_p\}$ for the parent $p$ of $\{x_i,\overline{x}_i\}$ in $T$ and embed this edge in $f_i$. Respectively, 
we add an edge $\{\overline{x}_i,v_p\}$  and embed this edge in $f_i$ if $x_i=false$. Denote the obtained graph by $\tilde{G}'$. By exactly the same arguments as for 
the proof of the inequality $\diam(\hat{G}')\leq d$, we have that $\diam(\tilde{G}')\leq d$.

Suppose now that $(\hat{G},k,d)$ is a yes-instance of {\sc  BFPDC}. Let $A$ be a set of edges such that the graph $\tilde{G}'$ obtained from $\tilde{G}$ by the addition of $A$ has diameter at most $d$. Because only the faces $f_1,\ldots,f_n$  have degree 5 and all other faces have degree 3, each edge of $A$ has its end-vertices in the boundary of some $f_i$ and is embedded in this face. 
Because $k=1$, at most one edge of $A$ is embedded in $f_i$ for $i\in\{1,\ldots,n\}$.  Let $v_r'$ be the pole of the mast rooted in $v_r$. Because $\diam(\tilde{G}')\leq d$, for any $u\in V(\tilde{G}')$, $\dist_{\tilde{G}'}(v_r',u)\leq d$ and, in particular, it holds for poles of other masts.
Consider masts rooted in $w_i$ for $e=\{x_i,\overline{x}_i\}\in L$. For a mast rooted in $w_i$, denote by $w_i'$ its pole. 
Because $\dist_{\tilde{G}'}(v_r',w_i\rq{})\leq d$, by the same arguments that were used above in the proof of the {\sf NP}-hardness of {\sc BPDC}, we obtain that it implies that for each $h\in\{1,\ldots,n\}$, either
$\{x_i,v_p\}\in A$ or $\{\overline{x}_i,v_p\}\in A$ where $p$ is the parent of $\{x_h,\overline{x}_h\}$ in $T$. 
For $h\in \{1,\ldots,n\}$, we let the variable $x_h=true$ if $\{\overline{x}_i,v_p\}\in A$ and $x_h=false$ otherwise. 
To prove that this assignment satisfies $\phi$, we again use the same arguments as above: it follows from the fact that for 
each clause $C_j$, $\dist_{\tilde{G}'}(v_r',C_j')\leq d$ where $C_j'$ is the pole of the mast rooted in the vertex $C_j$.

To complete  the proof of the {\sf NP}-hardness of {\sc BPDC}, it remains to observe that $\tilde{G}$ can be constructed in polynomial time.
\end{proof}

We proved that {\sc BPDC} is {\sf NP}-complete for 3-connected planar graphs. By the Whitney's theorem (see, e.g., \cite{Diestel12}), any two plane embeddings of a 3-connected plane graphs are equivalent. It gives the following corollary.

\begin{center}
\fbox{\small\begin{minipage}{12.7cm}
\noindent{\sc Bounded Budget Planar Diameter Completion}\\
{\sl Input}: A planar graph $G$, non-negative integers $k$ and $d$.\\
{\sl Output}:  Is it possible to obtain a planar graph $G'$ of diameter at most $d$ from $G$ by adding at most $k$ edges?
\end{minipage}}
\end{center}

\begin{corollary}\label{cor:planar-NPc}
\textsc{ Bounded Budget Planar Diameter Completion} is {\sf NP}-complete for 3-connected planar graphs.
\end{corollary}

\section{Discussion} 

We remark that our algorithm still works for the  classic 
{\sc PDC}  problem 
when the  face-degree of the input graph is bounded. For this we define the following problem:
\begin{center}
\fbox{\small\begin{minipage}{12.7cm}
\noindent{\sc Bounded Face BDC (FPDC)}\\
{\sl Input}: a plane graph $G$ with face-degree at most $k\in \mathbb{N}_{\geq 3}$,  and $d\in \mathbb{N}$ \\
{\sl Question}: is it possible to add edges in $G$ such that the resulting embedding remains plane and has diameter  at most $d$?
\end{minipage}}
\end{center}
We directly have the following corollary of Theorem~\ref{one}.

\begin{theorem}
\label{two}
It is possible to construct an $O(n^{3})+2^{2^{O((kd)\log d)}} \cdot n$-step algorithm for {\sc FPDC}.
\end{theorem}

To construct an {\sf FPT}-algorithm for {\sc PDC} when parameterized by $d$
remains an insisting open problem. The reason why our approach 
does not apply (at least directly) for {\sc PDC}  is that, 
as long as a completion may  add an arbitrary number of edges in each face,
we cannot guarantee  that our dynamic programming algorithm will be applied 
on a graph of bounded branchwidth. We believe that our approach and, in particular, the 
machinery of our dynamic programming algorithm,
might be useful for further 
investigations on this problem.

All the problems in this paper are defined on plane graphs. However, one may
also consider the ``non-embedded'' counterparts of the  problems 
{\sc PDC} and {\sc BPDC} by asking that their input is a planar combinatorial graphs (without a particular embedding). Similarly, such a counterpart can also be defined 
for the case of {\sc BFPDC} if we ask whether the completion has an embedding with at most 
$k$ new edges per face. Again, all these parameterized problems are known
to be (non-constructively) in  {\sf FPT}, because of the results in~\cite{RobertsonS-XX,RobertsonS95b}. However,  our approach fails to design the corresponding 
algorithms as it strongly requires an embedding of the input graph. For this reason 
we believe that even the non-embedded versions of {\sc BPDC}  and {\sc BFPDC} 
are as challenging as the general {\sc Planar Diameter Completion} problem.\\

\noindent\textbf{Acknowledgement}. We would like to thank the anonymous referees of an earlier version of this paper for their  remarks and suggestions that improved the presentation of the paper.

%

\begin{thebibliography}{10}

\bibitem{BoyerM04}
John~M. Boyer and Wendy~J. Myrvold.
\newblock On the cutting edge: Simplified ${O}(n)$ planarity by edge addition.
\newblock {\em J. Graph Algorithms Appl.}, 8(2):241--273, 2004.

\bibitem{ChatzidimitriouGMRTZ15fixe}
Dimitris Chatzidimitriou, Archontia~C. Giannopoulou, Spyros Maniatis, Clément
  Requilé, Dimitrios~M. Thilikos, and Dimitris Zoros.
\newblock Fixed parameter algorithms for completion problems on planar graphs.
\newblock Manuscript, 2015.

\bibitem{Chung87}
F.~R.~K. Chung.
\newblock Diameters of graphs: Old problems and new results.
\newblock {\em Congressus Numerantium}, 60, 1987.

\bibitem{Courcelle97}
Bruno Courcelle.
\newblock The expression of graph properties and graph transformations in
  monadic second-order logic.
\newblock {\em Handbook of Graph Grammars}, pages 313--400, 1997.

\bibitem{DejterF93impr}
Italo~J. Dejter and Michael~R. Fellows.
\newblock Improving the diameter of a planar graph.
\newblock Manuscript, may 1993.

\bibitem{Diestel12}
Reinhard Diestel.
\newblock {\em Graph Theory, 4th Edition}, volume 173 of {\em Graduate texts in
  mathematics}.
\newblock Springer, 2012.

\bibitem{DornPBF10effi}
Frederic Dorn, Eelko Penninkx, Hans~L. Bodlaender, and Fedor~V. Fomin.
\newblock Efficient exact algorithms on planar graphs: Exploiting sphere cut
  decompositions.
\newblock {\em Algorithmica}, 58(3):790--810, 2010.

\bibitem{DowneyF13fund}
Rodney~G. Downey and Michael~R. Fellows.
\newblock {\em Fundamentals of Parameterized Complexity}.
\newblock Texts in Computer Science. Springer, 2013.

\bibitem{GaoHN13}
Yong Gao, Donovan~R. Hare, and James Nastos.
\newblock The parametric complexity of graph diameter augmentation.
\newblock {\em Disc. Appl. Math.}, 161(10-11):1626--1631, 2013.

\bibitem{GuTa08}
Qian-Ping Gu and Hisao Tamaki.
\newblock {Optimal branch-decomposition of planar graphs in $O(n^3)$ time}.
\newblock {\em ACM Transactions on Algorithms}, 4(3), 2008.

\bibitem{GuT10impr}
Qian-Ping Gu and Hisao Tamaki.
\newblock Improved bounds on the planar branchwidth with respect to the largest
  grid minor size.
\newblock In {\em Algorithms and Computation - 21st International Symposium,
  (ISAAC 2010)}, pages 85--96, 2010.

\bibitem{HopcroftT74}
John~E. Hopcroft and Robert~Endre Tarjan.
\newblock Efficient planarity testing.
\newblock {\em J. ACM}, 21(4):549--568, 1974.

\bibitem{KoutsonasT10plan}
Athanassios Koutsonas and Dimitrios~M. Thilikos.
\newblock Planar feedback vertex set and face cover: Combinatorial bounds and
  subexponential algorithms.
\newblock {\em Algorithmica}, 60(4):987--1003, 2011.

\bibitem{ChungLMCTSL92}
Chung-Lun Li, S.~Thomas McCormick, and David Simchi-Levi.
\newblock On the minimum-cardinality-bounded-diameter and the
  bounded-cardinality-minimum-diameter edge addition problems.
\newblock {\em Oper. Res. Lett.}, 11(5):303--308, 1992.

\bibitem{Lichtenstein82}
David Lichtenstein.
\newblock Planar formulae and their uses.
\newblock {\em SIAM Journal on Computing}, 11(2):329--343, 1982.

\bibitem{RobertsonS95b}
Neil Robertson and Paul~D. Seymour.
\newblock Graph minors {XIII}. {T}he disjoint paths problem.
\newblock {\em J. Comb. Theory, Ser. B}, 63(1):65--110, 1995.

\bibitem{RobertsonS95-XIII}
Neil Robertson and Paul~D. Seymour.
\newblock {Graph Minors. {XIII}. {T}he disjoint paths problem}.
\newblock {\em J. Combin. Theory, Ser. B}, 63(1):65--110, 1995.

\bibitem{RobertsonS-XX}
Neil Robertson and Paul~D. Seymour.
\newblock Graph minors. {XX}. {W}agner's conjecture.
\newblock {\em J. Combin. Theory Ser. B}, 92(2):325--357, 2004.

\bibitem{SchooneBL87}
Anneke~A. Schoone, Hans~L. Bodlaender, and Jan van Leeuwen.
\newblock Diameter increase caused by edge deletion.
\newblock {\em Journal of Graph Theory}, 11(3):409--427, 1987.

\bibitem{SeymourT94call}
Paul~D. Seymour and Robin Thomas.
\newblock Call routing and the ratcatcher.
\newblock {\em Combinatorica}, 14(2):217--241, 1994.

\end{thebibliography}
%
%

\end{document}